\documentclass{article}

\usepackage{amsmath}
\usepackage{amssymb}
\usepackage{amsthm}

\usepackage{fontspec}
\usepackage{graphicx}
\newcommand{\poly}{\text{poly}}
\usepackage{nicematrix}
\usepackage[braket, qm]{qcircuit}

\usepackage{multirow}
\usepackage{booktabs}
\usepackage{array}
\usepackage{makecell}

\usepackage{geometry}
\usepackage{indentfirst}

\usepackage{pgfplots}
\pgfplotsset{compat=1.18}
\usepackage{subcaption} 

\usepackage{booktabs,multirow,siunitx}
\usepackage{authblk} 
\usepackage[colorlinks=true, urlcolor=blue]{hyperref} 

\newcommand{\sqlabel}[3]{%
  \hbox{\kern .5em
    \vbox{%
      \hbox{\scriptsize #1}%
      \vskip -0.25ex
      \hbox{$#2$}%
      \vskip -0.25ex
      \hbox{\scriptsize #3}%
    }%
  \kern .5em}%
}

\newcommand{\SQInOut}{\sqlabel{IN}{Sq}{OUT}}
\newcommand{\SQOutIn}{\sqlabel{OUT}{Sq}{IN}}

\newcommand{\SQInOutInv}{\sqlabel{IN}{Sq^{-1}}{OUT}}
\newcommand{\SQOutInInv}{\sqlabel{OUT}{Sq^{-1}}{IN}}

\newtheoremstyle{bracket}
  {3pt}   %
  {3pt}   %
  {\itshape}  %
  {}      %
  {\bfseries} %
  {.}     
  {0.5em} 
  {\thmname{#1}\ \thmnumber{#2}\ \thmnote{ [#3]}}

\theoremstyle{bracket}
\newtheorem{theorem}{Theorem}[section]
\newtheorem{lemma}[theorem]{Lemma}
\newtheorem{corollary}{Corollary}[section]

\newtheoremstyle{bracketdef}
  {3pt}{3pt}
  {}        %
  {}        
  {\bfseries}
  {.}
  {0.5em}
  {\thmname{#1}\ \thmnumber{#2}\ \thmnote{ [#3]}}

\theoremstyle{bracketdef}
\newtheorem{definition}{Definition}[section]

\begin{document}

\title{Space-Optimized and Experimental Implementations of Regev's Quantum Factoring Algorithm}

\author[1]{Wentao Yang}

\author[2]{Bao Yan}

\author[1,3]{Muxi Zheng}

\author[1]{Quanfeng Lu}

\author[4]
{Shijie Wei\thanks{\href{mailto:weisj@baqis.ac.cn}{weisj@baqis.ac.cn}}}

\author[4,1,5,6]
{Gui-Lu Long\thanks{\href{mailto:gllong@tsinghua.edu.cn}{gllong@tsinghua.edu.cn}}}

\affil[1]{State Key Laboratory of Low-Dimensional Quantum Physics and Department of Physics, Tsinghua University, Beijing 100084, China}

\affil[2]{State Key Laboratory of Mathematical Engineering and Advanced Computing, Zhengzhou 450001, China}

\affil[3]{Centre for Quantum Technologies, National University of Singapore 117543, Singapore}

\affil[4]{Beijing Academy of Quantum Information Sciences, Beijing 100193, China}

\affil[5]{Frontier Science Center for Quantum Information, Beijing 100084, China}
\affil[6]{Beijing National Research Center for Information Science and Technology, Beijing 100084, China}

\date{} %

\maketitle%

\begin{abstract}

The integer factorization problem (IFP) underpins the security of RSA, yet becomes efficiently solvable on a quantum computer through Shor's algorithm. Regev's recent high-dimensional variant reduces the circuit size through lattice-based post-processing, but introduces substantial space overhead and lacks practical implementations. Here, we propose a qubit reuse method by intermediate-uncomputation that significantly reduces the space complexity of Regev's algorithm, inspired by reversible computing. Our basic strategy lowers the cost from \( O(n^{3/2}) \) to \( O(n^{5/4}) \), and refined strategies achieve \( O(n \log n) \) which is  a  space lower bound within this model.
Simulations demonstrate the resulting time-space trade-offs and resource scaling. Moreover, we construct and compile quantum circuits that factor \( N = 35 \), verifying the effectiveness of our method through noisy simulations. A more simplified experimental circuit for Regev's algorithm is executed on a superconducting quantum computer, with lattice-based post-processing successfully retrieving the factors. These results advance the practical feasibility of Regev-style quantum factoring and provide guidance for future theoretical and experimental developments.

\end{abstract}

\section{Introduction}


The integer factorization problem (IFP) is a fundamental challenge in number theory and forms the basis of modern asymmetric cryptography, most notably RSA, whose security relies on the presumed classical intractability of factoring large integers. This hardness stems from the absence of any known classical polynomial-time algorithm. In 1994, Shor introduced a quantum algorithm that solves the IFP in polynomial time~\cite{Shor94}, providing an exponential speedup over classical approaches and establishing one of the most prominent examples of quantum advantage.

Subsequent research has focused on optimizing Shor’s algorithm to reduce the required number of qubits and circuit size. These theoretical improvements have primarily targeted either space requirements or constant-factor reductions in time~\cite{Beckman96,VBE96,Bea03,Haner17,Gidney17,Gid21}, while the asymptotic time complexity remains \( O(n^2 \log n) \) with fast integer multiplication or \( O(n^3) \) with standard multiplication. Small-scale experimental demonstrations, such as \( 15 \) and \( 21 \) on various quantum hardware platforms~\cite{Vandersypen01,Lucero2012,Martin-Lopez2012,Monz16,Amico2019}, have also verified the correctness of the algorithm and shown its feasibility in constrained quantum settings, paving the way for larger-scale implementations. Several other heuristic quantum factoring approaches have been experimentally used to factor comparatively large integers~\cite{Peng2008Adiabatic,Xu2012Factorization,yan2022factoring,Anschuetz2019VQF,Karamlou2021VQF}, yet Shor’s algorithm remains the main focus because of its clear complexity guarantees and provable exponential speedup.

In 2023, Regev proposed a high-dimensional generalization of Shor’s algorithm~\cite{Regev23} that replaces one-dimensional order finding with a higher-dimensional structure and employs lattice-based post-processing. This modification reduces the circuit size by a factor of \(O(n^{1/2})\), leveraging the efficiency of multiplying many small integers rather than performing large modular multiplications. However, the original construction introduces a large space overhead and reduces its practicality.\footnote{Another issue comes from the heuristic assumption, the validity of which was partially addressed in subsequent work~\cite{Pilatte24}. While Ekera \emph{et al.}~\cite{Ekera23} extended the overall framework to the discrete logarithm problem, thereby increasing its theoretical robustness and applicability.}Ragavan \emph{et al.} later reduced the space complexity via a Fibonacci-based modular-exponentiation scheme~\cite{Ragavan23} and then extended their structure to optimize the constant factors~\cite{Ragavan24}. Beyond this specific optimization, it remains important to explore more direct and general space-reduction strategies. Furthermore, to date, no complete gate-level circuit for Regev’s algorithm has been constructed, and no implementation on real quantum hardware has demonstrated even for the smallest nontrivial instance.

To address the above open question, we propose a space-optimized version of Regev’s algorithm~(SORA) for integer factorization. We first revisit Regev’s construction and explicitly identify repeated squaring in the modular exponentiation as the main source of its increased space complexity. We then introduce a direct space-reduction method that reuses ancilla qubits via intermediate uncomputation, rather than the Fibonacci-based strategy of Ragavan \emph{et al.}~\cite{Ragavan23}. A simple version of our approach reduces the space complexity from \(O(n^{3/2})\) to \(O(n^{5/4})\), and more refined strategies—drawing inspiration from the reversible pebble game in classical reversible computing~\cite{Ben89}—achieve a further reduction to \(O(n \log n)\). We prove that \(O(n \log n)\) is a space lower bound within this framework and show that it is attainable, thereby establishing space optimality for our approach. We also present numerical results for different strategies under our method, illustrating the trade-off between time and space resources.

Next, we design and compile proof-of-concept circuits for factoring the smallest nontrivial instance \( N = 35 \) using both Regev's approach and our method. We first evaluate the performance of our optimized algorithm through noisy simulations. The results indicate that our method reduces qubit consumption and improves the performance of the circuit in the presence of noise. Finally, we construct a further simplified circuit suitable for current quantum hardware to demonstrate Regev's approach to factorization on a real quantum device. Applying the lattice-based post-processing then yields the correct factorization.

In the process of constructing these circuits, we adopt several compilation techniques that were used in experimental implementations of Shor's algorithm~\cite{Monz16,Lanyon07} and proved particularly effective for proof-of-principle demonstrations. These techniques yield varying degrees of simplification to our circuits, significantly reducing both the number of qubits and gates, thereby enabling efficient numerical simulation as well as execution on real quantum devices. 


Overall, this work provides an alternative, uncomputation-based space-reduction method for Regev’s algorithm and gives a space-optimal strategy within this framework, while also offering resource-scaling data that can guide studies on larger instances. Furthermore, our proof-of-concept implementation for factoring \( N = 35 \) demonstrates the effectiveness of the proposed space-optimization method. We also show that, when combined with careful compilation techniques, Regev’s algorithm can be executed on currently available quantum hardware. These results help narrow the gap between Regev’s theoretical proposal and practical quantum implementations and indicate a feasible direction for subsequent theoretical and experimental studies. Moreover, our uncomputation-based space-reduction method is broadly applicable and can be extended to other quantum algorithms, providing a new pathway toward practical quantum computation.

This paper is organized as follows. Section~\ref{Framework} first reviews Regev’s algorithm and then introduces our optimized quantum algorithm  by space-reduction method. Section~\ref{Performance}  describes our proof-of-concept implementation of the \(N = 35\) instance and reports numerical simulation and experimental results.  Section~\ref{Disc} concludes the paper. Additional technical details appear in the Appendix.

\section{Framework of a space-optimized quantum factoring algorithm}\label{Framework}

In this section, we first briefly review Regev's algorithm. We then propose a method that reduces its space complexity, yielding a space-optimized quantum factoring algorithm. We compare this algorithm with the original Regev's algorithm to analyze the resulting time-space trade-off, and we also compare it with Shor's algorithm. Finally, we illustrate our approach with several concrete instances.

\subsection{Review of Regev's algorithm}

Regev's algorithm can efficiently solve the following integer factoring problem: Given an \( n \)-bit odd composite integer \( N = pq \), determine its nontrivial factors \( p \) and \( q \). It can be thought of as a multidimensional analogue of Shor’s algorithm, which shares some similarities with it. 
Regev's algorithm is executed by reducing the integer factoring problem to high-dimensional order-finding, a task that can be efficiently performed by a quantum subroutine with high probability. The subroutine can also be viewed as a lattice problem, which is described below.

Let \( b_1, b_2, \dots, b_d \) be small integers coprime to \( N \) (e.g., \( b_i \) is the \( i \)-th prime number), and define \( a_i = b_i^2 \). If we can find a solution \( (z_1, z_2, \dots, z_d) \) that satisfies
\begin{align}
\label{eq:solution}
    \prod_{i=1}^d a_i^{z_i} \equiv \prod_{i=1}^d (b_i^{z_i})^2 \equiv 1 \pmod{N},
\end{align}
it is expected \footnote{One might intuitively assume the probability to be constant, but this is difficult to prove. The requirement can be relaxed such that there exists at least one vector \( (z_1, \dots, z_d) \) satisfying \( b \not\equiv \pm 1 \pmod{N} \) with bounded norm at most \( T = \exp(O(n/d)) \). This heuristic assumption enables Regev's algorithm to achieve factorization with constant probability~\cite{Regev23}, and a weaker version of this condition has been proved in~\cite{Pilatte24}.} that \( b = \prod_{i=1}^d b_i^{z_i} \not\equiv \pm 1 \pmod{N} \). This facilitates factoring \( N \), as \( (b+1)(b-1) \equiv 0 \pmod{N} \), and computing \( \gcd(b \pm 1, N) \) yields a nontrivial factor of \( N \).  

The solution \( (z_1, z_2, \dots, z_d) \) can be viewed as a high-dimensional order\footnote{Unlike Shor’s algorithm, computing the order of the square number \( a_i \) here avoids the difficulties in factoring when the order of \( b_i \) is odd.}, but it is not unique. Specifically, if both \( \mathbf{x}_1 \) and \( \mathbf{x}_2 \) are solutions satisfying the condition~\eqref{eq:solution}, then any integer linear combination \( k_1 \mathbf{x}_1 + k_2 \mathbf{x}_2 \) (where \( k_1, k_2 \in \mathbb{Z} \)) also constitutes a solution. This reveals a lattice structure of the solution set, which can be formally defined as:
\begin{align}
\label{eq:lattice_def}
\mathcal{L} &=
\left\{ (z_1, \dots, z_d) \in \mathbb{Z}^d \;\middle|\; \left( \prod_{i=1}^d b_i^{z_i} \right)^2 \equiv 1 \pmod{N} \right\}
\subset \mathbb{Z}^d.
\end{align}
To obtain a nontrivial solution, we define the sublattice \( \mathcal{L}_0 \) as:
\begin{align*}
\mathcal{L}_0 =
\left\{ (z_1, \dots, z_d) \in \mathbb{Z}^d \;\middle|\; \prod_{i=1}^d b_i^{z_i} \equiv \pm 1 \pmod{N} \right\} \subseteq \mathcal{L}.
\end{align*}
Our goal is to find a vector \( \mathbf{z} \in \mathcal{L} \setminus \mathcal{L}_0 \). An illustration of the lattice \( \mathcal{L} \) and its nontrivial elements is shown in Figure~\ref{fig:1}. Fortunately, Regev's algorithm guarantees success with probability at least \( 1/4 \) for finding such an element.

The core quantum subroutine in Regev's algorithm is high-dimensional order finding within the lattice framework, which is formally described below.

\subsubsection*{High-dimensional order-finding subroutine}

\begin{enumerate}

\item \textbf{Quantum state initialization}:

Similar to Shor's algorithm, we prepare two quantum registers to initialize the quantum state, setting \( d \) to approximately \(\sqrt{n}\) and \( D \), a power of 2, to order \(O(n^{1/4}\exp(C\sqrt{n})) \).\footnote{A detailed discussion of parameter choices is provided at Appendix~\ref{D}.} The initial state is prepared as:
\begin{equation}
\label{eq:initialstate}
    \sum_{z \in \{-D/2, \dots, D/2-1\}^d} \rho_R(z) \ket{z} \ket{0},
\end{equation}
where the first register \(\ket{z}\) is divided into \( d \) blocks, each containing \(\log D\) qubits, to encode the binary representation of \( z = (z_1, \dots, z_d) \), where each component \( z_i \in \{-D/2, \dots, D/2-1\} \). The second register is initialized to \(\ket{0}\) and will hold the result of modular exponentiation in the next step.

The discrete Gaussian distribution is commonly used in classical cryptography, defined as:

\begin{equation}
    \label{eq:gaussian}
    \rho_R(z) = \exp\left(-\pi \|z\|^2 / R^2\right), \quad \text{for } z \in \mathbb{Z}^d.
\end{equation}

Regev proposed that precise state preparation is not required. It suffices to initialize the \( O(\log d) \) most significant qubits of each block to approximate the Gaussian distribution, while the remaining qubits can be set to a uniform superposition using Hadamard gates.

\item \textbf{\(d\)-dimensional modular exponentiation}: 

As the next step, we compute the value \( \prod_{i} a_i^{z_i + D/2} \pmod{N} \) in the second register (we add \( D/2 \) for convenience, ensuring \( z_i + D/2 \in \{0, \ldots, D-1\} \)):
\begin{equation}
    \label{eq:modexp}
    \ket{z} \ket{0} \to \ket{z} \ket{\prod_{i} a_i^{z_i + D/2} \pmod{N}}.
\end{equation}

This can be achieved by applying a series of controlled-\(U\) operators \( U_{a_i}, U_{a_i}^2, U_{a_i}^4, \dots \) to compute \( a_i^{z_i + D/2} \pmod{N} \) for each \( a_i \) in the second register, as done in Shor's algorithm, referred to as the precomputation method. However, this approach needs precomputing \( a_i^{2^j} \) and multiplying it into the second register, requiring a total of \( d \log D = O(n) \) \(n\)-bit multiplication operations, similar to Shor's algorithm. To reduce computational complexity, it is more efficient to first compute the product of small integers and apply the square-and-multiply algorithm:

First, let \( z_{ij} \) denote the \( j \)-th bit of \( z_i + D/2 \), with \( j = 0 \) being the most significant bit. Initializing the second register to \( \ket{1} \), we perform the following for \( j = 0, \dots, \lfloor \log_2 (D-1) \rfloor \): square the value in the second register, compute the product of small integers \( \prod_{i} a_i^{z_{ij}} \pmod{N} \), and multiply the second register by the result.

To compute \( \prod_{i} a_i^{z_{ij}} \pmod{N} \), Regev proposed a binary tree-based multiplication method, which recursively performs multiplication starting from pairs of small primes. This leverages the fact that the \( a_i \) are all small \( O(\log d) \)-bit integers, thereby reducing time complexity.

\item \textbf{\(d\)-dimensional Fourier transform}: 

We apply the \(d\)-dimensional quantum Fourier transform over \( \mathbb{Z}_D^d \) to the \( \ket{\mathbf{z}} \) register, which is equivalent to applying a one-dimensional Fourier transform over \( \mathbb{Z}_D \) to each dimension independently.

\item \textbf{Measurement}: 

We measure the first register and divide the result by \( D \), obtaining an output vector in \( \{0, 1/D, \dots, (D-1)/D\}^d \), which provides a value for further processing. The overall quantum circuit for Regev's algorithm can be represented in figure~\ref{fig:algorithm}

\end{enumerate}

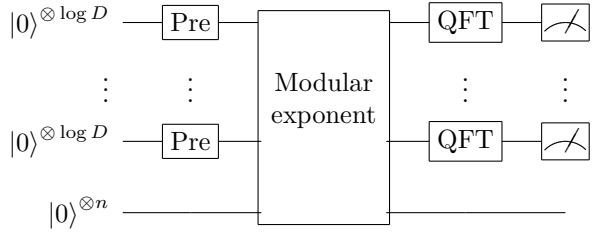
\begin{figure}[htbp]
  \centering
  \[
  \Qcircuit @C=1.5em @R=1.5em {
      \lstick{\ket{0}^{\otimes \log D}} & \gate{\text{Pre}} & \multigate{3}{\text{exponent}} & \gate{\text{QFT}} & \meter{} \\
      \lstick{\vdots}                   & \vdots            & \text{Modular}                 & \vdots            & \vdots \\
      \lstick{\ket{0}^{\otimes \log D}} & \gate{\text{Pre}} & \ghost{\text{Modular}}         & \gate{\text{QFT}} & \meter{} \\
      \lstick{\ket{0}^{\otimes n}}      & \qw               & \ghost{\text{Modular}}         & \qw               & \qw
  }
  \]
  \caption{Quantum circuit for the Regev's algorithm.}
  \label{fig:algorithm}
\end{figure}

\subsubsection*{Post-processing by lattice basis reduction}
After applying the \(d\)-dimensional quantum Fourier transform to the first register, we obtain a state that is approximately a superposition of vectors close to elements of the dual lattice \( \mathcal{L}^*/\mathbb{Z}^d \), where the dual lattice is defined as:
\begin{equation}
\label{eq:dual_lattice}
\mathcal{L}^\ast = \left\{ \mathbf{x} \in \mathbb{R}^d : \langle \mathbf{x}, \mathbf{z} \rangle \in \mathbb{Z}, \ \forall \mathbf{z} \in \mathcal{L} \right\} \supseteq \mathbb{Z}^d.
\end{equation}

Therefore by repeating the quantum procedure \( d+4 \) times, we obtain \( d+4 \) samples of \( \mathcal{L}^*/\mathbb{Z}^d \) with noise, as the measurement results are restricted to the grid \( \{0, 1/D, \dots, (D-1)/D\}^d \). Our goal is to recover elements of the lattice \( \mathcal{L} \) from these noisy samples of the dual lattice \( \mathcal{L}^*/\mathbb{Z}^d \). Regev proposes a lattice-based reconstruction method by generating the following lattice:

\NiceMatrixOptions{cell-space-limits = 5pt}
\[
B = 
\begin{pNiceArray}{c|c}[margin]
I_{d \times d} & 0 \\
\hline
S w_1 & \\
\Vdots & S I_{(d+4) \times (d+4)} \\
S w_{d+4} & 
\end{pNiceArray} \; .
\]

where \( I_{d \times d} \) is the \( d \times d \) identity matrix, \( S \) is a scaling factor, and \( w_i \) are vectors derived from the noisy samples. Under appropriate parameter selection, applying the LLL algorithm to find relatively short lattice vectors ensures that the first \( d \) components are the desired elements of \( \mathcal{L} \) with probability at least \( 1/4 \).
This success probability is ensured by appropriate parameter selection. We list the parameter selection ranges in Appendix~\ref{D}, which are applicable to both the mathematical proof and asymptotic analysis.

\subsubsection*{An overview of Regev's computational complexity}

The time complexity of Regev's algorithm is dominated by approximately \( \log D \) modular multiplication and squaring operations on \( n \)-bit numbers, each requiring a circuit of size \( O(n \log n) \).\footnote{For convenience, we adopt Fast Multiplication algorithm with circuit size \( O(n \log n) \) for multiplying \(n\)-bit numbers and performing squarings.} Following Regev's parameter choice, taking \( \log D = O(\sqrt{n}) \) and \( d \approx \sqrt{n} \) yields an overall time complexity of \( O(n^{3/2} \log n) \). A detailed analysis of the complexity of Regev's algorithm is presented in Appendix~\ref{B}.

A significant challenge of Regev's algorithm arises from its space complexity. The modular exponentiation step uses \( O(n^{3/2}) \) qubits, which dominates the overall space complexity. This is significantly larger than that of Shor's algorithm, which requires \( O(n \log n) \) qubits.

The increase in space complexity is due to the irreversibility of the squaring operations. In the square-and-multiply procedure, we cannot implement an in-place squaring operation \( \ket{x} \xrightarrow{} \ket{x^2 \pmod{N}} \); instead, we must use out-of-place operations \( \ket{x}\ket{0} \xrightarrow{} \ket{x}\ket{x^2 \pmod{N}} \) and store all intermediate results \( \ket{x} \) after the squaring operations. These intermediate states are subsequently removed through uncomputation. Since \( \log D \) intermediate states are stored during the process, each requiring \( n \) qubits, this results in a total of \( O(n^{3/2}) \) ancilla qubits.

This space complexity significantly limits the practical potential of Regev's algorithm, as the number of qubits is a limiting factor in current quantum computing.\footnote{It should also be noted that, although the squaring operations can be replaced by classical precomputation, this removes the asymptotic time advantage of Regev's algorithm and substantially reduces its significance.} We therefore introduce an optimized quantum factoring algorithm to reduce the space complexity.


\subsection{Space-optimized Regev's algorithm}

In this section, we present a space-optimized quantum factoring algorithm that reduces the space complexity of Regev's algorithm by employing a set of qubit-reuse techniques.

\subsubsection{A method for qubit reuse: intermediate uncomputation}

An intuitive idea is to reuse the computational registers storing \( \ket{x} \) in order to reduce the number of ancilla qubits. However, directly erasing the intermediate states is equivalent to measuring them, which would destroy the coherent information that is essential for the algorithm. The idea of uncomputing intermediate states \( \ket{x^2 \pmod{N}} \) before completing the full computation to reuse qubits is feasible. However, this requires retaining the previous ancilla qubits \( \ket{x} \), which imposes certain restrictions.


Under these constraints, we introduce a method based on a technique that we term Intermediate Uncomputation. A related model, known as the reversible pebble game~\cite{Ben89}, has been studied in the field of classical reversible computing, with slight differences from our method. We show that, by using our method, the space complexity of Regev's algorithm can be reduced to \( O(n \log n) \). In the subsequent sections and appendices, we provide exact space and time costs, thereby illustrating the trade-off between time and space.



To illustrate the basic idea, suppose a quantum algorithm requires performing \( m \) out-of-place squaring operations, necessitating \( m \) additional registers.(For Regev's algorithm, we have \( m = \log D - 1 \)) The intermediate states are eventually uncomputed at the end. During the procedure, we can perform the out-of-place inverse squaring operations \( \ket{x}\ket{x^2 \pmod{N}} \xrightarrow{} \ket{x}\ket{0} \) to uncompute intermediate states, thereby freeing some registers. 

Since uncomputing via inverse squaring requires retaining the state of the previous register, we propose a simple and easy-to-implement strategy: after using \( k \) computational registers (in which some squaring operations are performed along with reversible in-place operations, such as modular multiplication in our case), we immediately reverse the circuit on the previous \( k-1 \) registers. This will uncompute the results of the squaring and multiplication operations stored in these \( k-1 \) registers, in reverse order. 

Thus, upon obtaining the final result, only \( \lceil (m+1)/k \rceil - 1\) intermediate states need to be stored. These intermediate states still require sequential uncomputation, which involves preserving the final result while running all the above circuits in reverse. In this process, at most \( \lceil (m+1)/k \rceil + k - 1 \) registers are used to store intermediate states, where \( \lceil x \rceil \) denotes the smallest integer not less than \( x \).

Clearly, when \( k \approx \sqrt{m} \), the space (required registers) is minimized at approximately \( 2\sqrt{m} \), while the time (counted by the number of squaring operations) increases by at most a constant factor not exceeding 2 (since the operations on each register are not scaled by more than a factor of 2). Under this strategy, the space complexity of Regev’s algorithm can be immediately reduced to \( O(n^{5/4}) \), while the asymptotic time complexity remains unchanged.

We term the above method the ``intermediate uncomputation'' algorithm (also referred to as ``partial uncomputation'') and refer to the specific strategy described above as the ``simple strategy''. 
In fact, this strategy is simple but not asymptotically optimal. To further reduce the space complexity, we introduce the binary recursive strategy. Its idea stems from the divide-and-conquer paradigm: if we can perform \( m \) squarings together with all required uncomputations using \( r \) registers, then we can perform \( 2m+1 \) squarings and all uncomputations at the cost of \(r + 1\) registers. In practice, this is achieved by first completing the first stage of \( m \) squarings together with their uncomputation, retaining the register that holds the final result, and then proceeding with the second stage. This yields approximately \( O(\log m )\) register consumption, at the expense of an increased time complexity of \( O(m^{\log_2 3}) \). It is proven that this strategy completes the computation using the minimal number of registers, achieving the lower bound on space complexity. By extending to arbitrary \( k \)-ary recursion, we can achieve a space complexity of \( O(\log m) \) and a time complexity of \( O(m^{1 + \epsilon}) \), where \( \epsilon \) is an arbitrarily small positive number. In the following sections, we will introduce different strategies and the complexity conclusions of different strategies and provide their specific time and space costs.

\subsubsection{Computational complexity} 

We summarize the key results below, with detailed derivations provided in Appendix~\ref{C}. Here, \( m \) denotes the number of squarings in the underlying mathematical task (for Regev's algorithm, \( m = \log D - 1 \)). The time complexity is measured in squaring operations of the quantum circuit (including the inverse of squaring and those arising from uncomputation), the space complexity is measured in the number of \( n \)-bit registers.


\begin{definition}[Direct Computation]
Without any optimization, we directly perform \( m \) out-of-place squaring operations and finally uncompute all intermediate registers.
\end{definition}

Direct computation is time-optimal but space-inefficient: it uses \( m+1 \) registers and minimizes squaring operations to \( 2m - 1 \), corresponding to Regev's original idea and representing the minimum time complexity. However, it is inefficient in terms of space usage. Therefore, we propose several strategies to reduce the space complexity, starting with the simple strategy introduced earlier:

\begin{corollary}[Simple Strategy (Space-Reduced, Time-Efficient)]
The number of registers required to compute \( m \) squarings can be reduced to \( O(\sqrt{m}) \), while the time complexity remains \( O(m) \), increasing by at most a constant factor not exceeding 2 compared to direct computation.
\end{corollary}

To further reduce the space complexity, we introduce the binary recursive strategy. It is proven that this strategy completes the computation using the minimal number of registers, achieving the lower bound on space complexity.


\begin{theorem}[Space Lower Bound]
Using \( n \) registers, one cannot complete the computation of \( 2^{n-1} \) squarings (together with their uncomputation). Consequently, for \( m \) squarings, at least \( \lceil \log_2(m+1) \rceil + 1 \) registers are required. This establishes a lower bound on the space complexity.
\end{theorem}

\begin{corollary}[Binary Recursion (Achieving the Space Lower Bound)]

The space complexity lower bound described above for \( m \) can be achieved using the binary recursion strategy, with the time complexity not exceeding \( O(m^{\log_2 3}) \).
\end{corollary}

The recursive strategy can be generalized from binary recursion to fixed \( k \)-ary recursion, which yields different asymptotic complexity conclusions. Allowing \( k \) to grow with \( m \), a setting we refer to as variable-arity recursion, leading to different asymptotic time--space trade-offs. For any fixed \( m \), these two viewpoints correspond to choosing an appropriate effective value of \( k \), and thus are equivalent descriptions of the same family of strategies.

\begin{corollary}[k-ary Recursion (Log-Space, Near-Linear Time)]

The asymptotic time complexity can be reduced to \( O(m^{1 + \epsilon}) \) using a \( k \)-ary recursion strategy, where \( \epsilon \) is an arbitrarily small positive constant, specifically \( \epsilon = \log_k (2k-1) - 1 \). At the same time, the space complexity is kept at \( O(\log m) \), with the constant factor in space approximately \( \frac{k-1}{\log_2 k} \).
\end{corollary}

\begin{corollary}[Variable-Arity Recursion (Linear Time, Sublinear Space)]
For any integer \( \ell \ge 1 \), there exists a variable-arity recursion strategy that computes \( m \) squarings using only \( O\!\bigl(\sqrt[\ell]{m}\bigr) \) registers, while keeping the time complexity at \( O(m) \). This strategy can be realized by choosing the recursion arity as \( k \approx \sqrt[\ell]{m+1} \).
\end{corollary}

Considering \( m = \log D - 1 \approx C\sqrt{n} \), the time and space complexities of our method with the simple, binary recursive, \( k \)-ary recursive, and variable-arity recursion strategies are summarized in Table~\ref{tab:comparison3}. Since a sufficiently small \( \epsilon \) leads to a significant increase in the constant factor of the \( O(\log m) \) space complexity, and choosing \( \ell \) very large similarly increases the constant factor in the \( O(m) \) time complexity, we need to perform a space--time trade-off in practical applications. In these situations, the simple strategy or a recursion of suitable order is sufficient, significantly reducing the space complexity compared to the original Regev’s algorithm, with time complexities of \( O(n^{3/2} \log n) \) or at most \( O(n^{1.79} \log n) \), both still remaining lower than that of Shor’s algorithm.

\begin{table}[htbp]
\centering
\begin{tabular}{|c|c|c|c|}
  \hline
  \multicolumn{2}{|c|}{\textbf{Algorithm}} & \textbf{Space Complexity} & \textbf{Time Complexity} \\
  \hline
  \multicolumn{2}{|c|}{Shor's~\cite{Shor94,Haner17}} & \( O(n \log n)^* \) & \( O(n^2 \log n) \) \\
  \hline
  \multicolumn{2}{|c|}{Regev's ~\cite{Regev23}} & \( O(n^{3/2}) \) & \( O(n^{1.5} \log n) \) \\
  \hline
  \multirow{4}{*}{Our} 
    & Simple (Time-Efficient) & \( O(n^{5/4}) \) & \( O(n^{1.5} \log n) \) \\
  \cline{2-4}
    & Binary recursion (Space Lower Bound) & \( O(n \log n) \) & \( O(n^{1.79} \log n) \) \\
  \cline{2-4}
    & \(k\)-ary recursion& \( O(n \log n) \) & \( O(n^{1.5 + \epsilon} \log n) \) \\
  \cline{2-4}
    & Variable-arity recursion & \( O\!\bigl(n^{1 + \frac{1}{2\ell}}\bigr) \) & \( O(n^{1.5} \log n) \) \\
  \hline
\end{tabular}
\caption{Comparison of the space and time complexities of different algorithms.}
{\footnotesize $^*$The \(\log n\) factor in Shor's algorithm comes from fast multiplication.}
\label{tab:comparison3}
\end{table}

\subsubsection{Resource usage and scaling of our theoretical method}

To reveal the asymptotic growth, we provide expressions for the algorithm's resource consumption and characterize how the computational cost scales with \( m = \log D - 1 \).

Here, a natural way to compare resource consumption is to count the number of computational registers (as a measure of space complexity) and the number of squarings or large-integer multiplications (as a measure of time complexity, noting that both are equal). This comparison method is effective if we assume that squarings and large-integer multiplications are the most time-consuming modules of the algorithm, which becomes increasingly accurate as the problem size grows.

For the simple strategy with block size \(k\), the space and time admit closed-form expressions:
\begin{align}
S_{\mathrm{simple}}(m,k)
  &= \left\lceil \frac{m+1}{k} \right\rceil + k - 1, \\
T_{\mathrm{simple}}(m,k)
  &= \Bigl( 2\Bigl\lceil \frac{m+1}{k} \Bigr\rceil - 1 \Bigr)(2k-1)
     + 2\Bigl( (m+1) - k \Bigl\lceil \frac{m+1}{k} \Bigr\rceil \Bigr)-4 .
\end{align}

For the \(k\)-recursive strategy, the concrete complexity depends on the chosen base unit \( x_0 \), which uses \( n_0 \) registers to compute \( x_0 \) squarings in time \( t_0 \). Fixing \( x_0, n_0, t_0 \) according to this base unit, and assuming that
\(\log_k\!\left((m+1)/x_0\right)\) is a positive integer, we obtain the following space cost and an upper bound on the time cost (proved by a constructive argument in Appendix~\ref{C}):
\begin{align}
S_{\mathrm{rec}}(m;k,x_0) &= n_0 + (k-1)\,\log_k\!\left(\frac{m+1}{x_0}\right), \\
T_{\mathrm{rec}}(m;k,x_0) &\le t_0 \left(\frac{m+1}{x_0}\right)^{\log_k(2k-1)}.
\end{align}
These expressions can be used as asymptotic estimates for the actual space and time costs. 

The theoretical derivations of the above formulas are given in Appendix~\ref{C}, and the resulting scaling with \( m \) is summarized in Figure~\ref{fig:scaling}. A side-by-side comparison for representative values of \( m \) appears in Table~\ref{tab:summary-m-direct-simple-binary}. Figure~\ref{fig:scaling}(a) shows that the register count \(S(m)\) for the direct strategy grows linearly with \(m\), while the simple and recursive strategies substantially reduce the space overhead. In contrast, Figure~\ref{fig:scaling}(b) shows that the time costs \(T(m)\) for Direct and Simple remain essentially linear in \(m\), whereas the \(k\)-ary recursive strategies follow power laws with larger exponents, initially revealing the time--space trade-off between reduced space and increased runtime.

    
\begin{figure}[htbp]
  \centering
  \begin{subfigure}{0.48\linewidth}
    \centering
    \includegraphics[width=\linewidth]{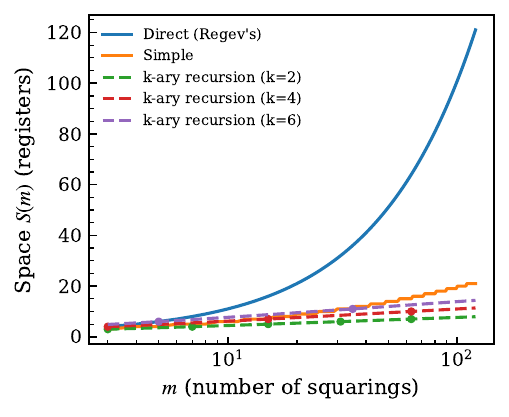} 
    \caption{}
    \label{fig:space-scaling}
  \end{subfigure}
  \hfill
  \begin{subfigure}{0.48\linewidth}
    \centering
    \includegraphics[width=\linewidth]{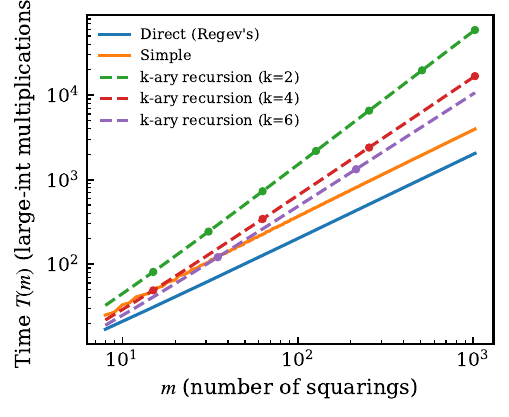} 
    \caption{}
    \label{fig:time-scaling}
  \end{subfigure}
\caption{Scaling of the space and time costs of the squaring subroutine in our optimized Regev's algorithm.
(a) Number of computational registers \( S(m) \) as a function of the number of squarings \( m = \log D - 1 \) for the direct strategy (original Regev construction), the simple intermediate-uncomputation strategy, and \( k \)-ary recursive strategies with \( k \in \{2,4,6\} \). (b) Corresponding time cost \(T(m)\) measured in large-integer squarings or multiplications. Here \(m\) denotes the number of squarings in Regev's algorithm, with \(m = \log D - 1 = C\sqrt{n} - 1\). In both panels, solid dots and solid lines represent the exact analytical results given by Eqs.~(7)--(10); dashed lines are analytic fits for the \( k \)-ary cases. }
  \label{fig:scaling}
\end{figure}


\begin{table}[htbp]
\centering
\begin{tabular}{@{} S[table-format=3.0]
S[table-format=3.0] S[table-format=3.0]
S[table-format=2.0] c
S[table-format=1.0] S[table-format=4.0] @{}}
\toprule
 & \multicolumn{2}{c}{Direct} & \multicolumn{2}{c}{Simple} & \multicolumn{2}{c}{Binary} \\
\cmidrule(lr){2-3}\cmidrule(lr){4-5}\cmidrule(lr){6-7}
{$m$} & {regs} & {mults} & {regs} & {mults} & {regs} & {mults} \\
\midrule
  3   &   4 &   5 &  3 &  5  & 3 &    5 \\
  7   &   8 &  13 &  5 & 17  & 4 &   19 \\
 15   &  16 &  29 &  7 & 45  & 5 &   65 \\
 31   &  32 &  61 & 11 & 107  & 6 &  211 \\
 63   &  64 & 125 & 15 & 221 & 7 &  665 \\
127   & 128 & 253 & 22 & 471 & 8 & 2059 \\
\bottomrule
\end{tabular}
\caption{Register number (space)  and large-integer multiplication counts (time) for different $m$ under three strategies.}
\label{tab:summary-m-direct-simple-binary}
\end{table}

\subsubsection{Time-space trade-offs}


We also summarize the time--space trade-offs achieved by our intermediate-uncomputation method and compare their time complexity with that of Shor's algorithm. As before, we measure space by the number of computational registers and time by the number of large-integer multiplications. Figure~\ref{fig:tradeoff-and-shor} presents these trade-offs and the comparison, panel~(a) shows the trade-off between time and space at a fixed number of squarings, while panel~(b) compares the scaling of the time cost with the problem size. Over the parameter ranges of interest, our optimized Regev-type schemes retain an asymptotic time advantage over Shor while requiring substantially fewer qubits than Regev's original construction. 




\begin{figure}[h!]
  \centering
  \begin{subfigure}{0.48\linewidth}
    \centering
    \includegraphics[width=\linewidth]{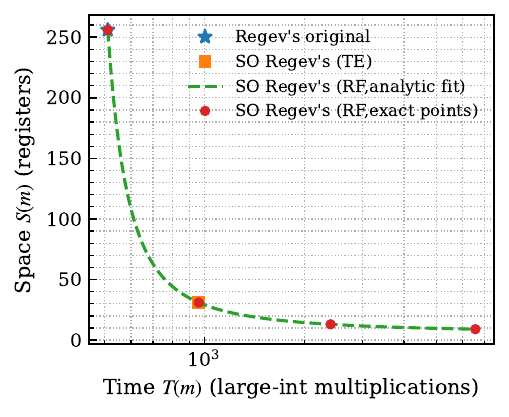}%
    \caption{}
    \label{fig:tradeoff-m255}
  \end{subfigure}
  \hfill
  \begin{subfigure}{0.48\linewidth}
    \centering
    \includegraphics[width=\linewidth]{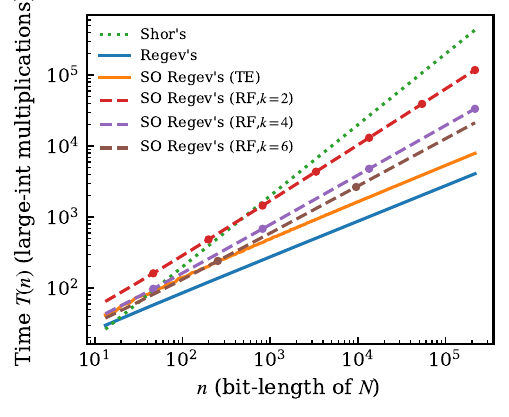}%
    \caption{}
    \label{fig:time-scaling-shor}
  \end{subfigure}
  \caption{Comparison of our optimized Regev’s algorithm with the original Regev’s algorithm and Shor’s algorithm. (a) Time--space trade-off at a fixed number of squarings $m=255$, showing the Direct strategy (Regev's original method), the simple strategy, and $k$-recursive strategies.(b) Scaling of the time cost (large-integer multiplications) with problem size, comparing Regev’s original algorithm, Shor’s algorithm~\cite{Haner17}, and our space-optimized Regev’s algorithm (denoted as “SO Regev’s” in the legend). Here TE stands for the time-efficient simple strategy, and RF,$k$ stands for refined $k$-recursive strategy. The Shor curve is based on the modular exponentiation using approximately $2n$ modular multiplications for an $n$-bit modulus. \protect\footnotemark\
  For Regev-type algorithms we choose a constant \(C=2.2\) as a representative choice following Ekera's analysis~\cite{Ekera25}. In practice, a lower value is often used.}
  \label{fig:tradeoff-and-shor}
\end{figure}

\subsubsection{Explicit examples for the intermediate uncomputation method}

Finally, we present specific examples, starting from small-scale instances, and provide a concrete quantum circuit to illustrate the effect of qubit reuse. A first nontrivial instance is the case when \( \log D = 4 \) and \( m = 3 \) squarings are needed, thereby requiring a total of 4 registers in the direct strategy.

One can uncompute the first register after the computation of the second register finishes, which resets the first register to \( \ket{0}^{n} \) and then reuses it to hold the result of the next squaring. The complete circuit diagram is shown in Figure~\ref{fig:uncompute1}.For comparison, the circuit for the direct strategy is shown in Figure~\ref{fig:uncompute0}. In the figures, \( U \) denotes the modular multiplication \( a_1^{z_{1j}} \cdot a_2^{z_{2j}} \) controlled by the two control registers \( z_{1j} \) and \( z_{2j} \), and \( Sq \) denotes the out-of-place squaring operation.


\footnotetext{The Shor curve shown here does not incorporate low-level engineering optimizations such as windowed arithmetic~\cite{Gid21}; it serves as a clean asymptotic reference and in practice often achievable with smaller values.}

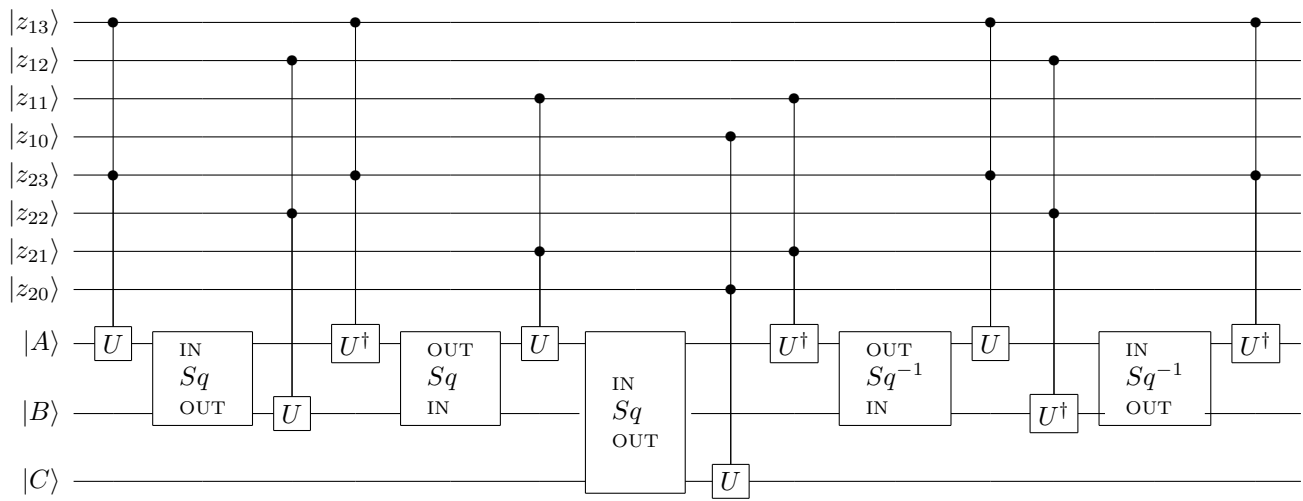
\begin{figure}[htbp]
  \centering
  \[
\Qcircuit @C=0.8em @R=1.2em{
\lstick{\ket{z_{13}}} & \ctrl{8} & \qw & \qw & \ctrl{8} & \qw & \qw & \qw & \qw & \qw & \qw & \ctrl{8} & \qw & \qw & \ctrl{8} & \qw \\
\lstick{\ket{z_{12}}} & \qw & \qw & \ctrl{8} & \qw & \qw & \qw & \qw & \qw & \qw & \qw & \qw & \ctrl{8} & \qw & \qw & \qw \\
\lstick{\ket{z_{11}}} & \qw & \qw & \qw & \qw & \qw & \ctrl{6} & \qw & \qw & \ctrl{6} & \qw & \qw & \qw & \qw & \qw & \qw \\
\lstick{\ket{z_{10}}} & \qw & \qw & \qw & \qw & \qw & \qw & \qw & \ctrl{7} & \qw & \qw & \qw & \qw & \qw & \qw & \qw \\
\lstick{\ket{z_{23}}} & \ctrl{4} & \qw & \qw & \ctrl{4} & \qw & \qw & \qw & \qw & \qw & \qw & \ctrl{4} & \qw & \qw & \ctrl{4} & \qw \\
\lstick{\ket{z_{22}}} & \qw & \qw & \ctrl{4} & \qw & \qw & \qw & \qw & \qw & \qw & \qw & \qw & \ctrl{4} & \qw & \qw & \qw \\
\lstick{\ket{z_{21}}} & \qw & \qw & \qw & \qw & \qw & \ctrl{2} & \qw & \qw & \ctrl{2} & \qw & \qw & \qw & \qw & \qw & \qw \\
\lstick{\ket{z_{20}}} & \qw & \qw & \qw & \qw & \qw & \qw & \qw & \ctrl{3} & \qw & \qw & \qw & \qw & \qw & \qw & \qw \\
\lstick{\ket{A}} & \gate{U} & \multigate{1}{\SQInOut} & \qw & \gate{U^\dagger} & \multigate{1}{\SQOutIn} & \gate{U} & \multigate{2}{\SQInOut} & \qw & \gate{U^\dagger} & \multigate{1}{\SQOutInInv} & \gate{U} & \qw & \multigate{1}{\SQInOutInv} & \gate{U^\dagger} & \qw \\
\lstick{\ket{B}} & \qw & \ghost{\SQInOut} & \gate{U} & \qw & \ghost{\SQOutIn} & \qw & \ghost{\SQInOutInv} & \qw & \qw & \ghost{\SQOutInInv} & \qw & \gate{U^\dagger} & \ghost{\SQInOut} & \qw & \qw \\
\lstick{\ket{C}} & \qw & \qw & \qw & \qw & \qw & \qw & \ghost{\SQInOut} & \gate{U} & \qw & \qw & \qw & \qw & \qw & \qw & \qw \\
} \]
  \caption{Circuit for the modular exponentiation part of our simple intermediate-uncomputation method for \( d=2 \), \( \log D = 4 \) and \( m=3 \) with 3 computational registers.}
  \label{fig:uncompute1}
\end{figure}

\begin{figure}[htbp]
  \centering
  
\[
\Qcircuit @C=1.2em @R=1.2em{
\lstick{\ket{z_{13}}} & \ctrl{8} & \qw & \qw & \qw & \qw & \qw & \qw & \qw & \qw & \qw & \qw & \ctrl{8} & \qw \\
\lstick{\ket{z_{12}}} & \qw & \qw & \ctrl{8} & \qw & \qw & \qw & \qw & \qw & \qw & \ctrl{8} & \qw & \qw & \qw \\
\lstick{\ket{z_{11}}} & \qw & \qw & \qw & \qw &  \ctrl{8} & \qw  & \qw &  \ctrl{8} &\qw & \qw & \qw & \qw  & \qw \\
\lstick{\ket{z_{10}}} & \qw & \qw & \qw & \qw & \qw & \qw &  \ctrl{8} &\qw & \qw & \qw & \qw & \qw & \qw \\
\lstick{\ket{z_{23}}} & \ctrl{4} & \qw & \qw & \qw & \qw & \qw & \qw & \qw & \qw & \qw & \qw & \ctrl{4} & \qw \\
\lstick{\ket{z_{22}}} & \qw & \qw & \ctrl{4} & \qw & \qw & \qw & \qw & \qw & \qw & \ctrl{4} & \qw & \qw & \qw \\
\lstick{\ket{z_{21}}} & \qw & \qw & \qw & \qw &  \ctrl{4} & \qw  & \qw &  \ctrl{4} &\qw & \qw & \qw & \qw  & \qw \\
\lstick{\ket{z_{20}}} & \qw & \qw & \qw & \qw & \qw & \qw &  \ctrl{4} &\qw & \qw & \qw & \qw & \qw & \qw \\
\lstick{\ket{A}} & \gate{U} & \multigate{1}{\SQInOut} & \qw & \qw & \qw & \qw & \qw & \qw & \qw & \qw & \multigate{1}{\SQInOutInv} & \gate{U^\dagger} & \qw \\
\lstick{\ket{B}} & \qw & \ghost{\SQInOut} & \gate{U} & \multigate{1}{\SQInOut} & \qw & \qw & \qw & \qw & \multigate{1}{\SQInOutInv} & \gate{U^\dagger} & \ghost{\SQInOutInv} & \qw & \qw \\
\lstick{\ket{C}} & \qw & \qw & \qw &  \ghost{\SQInOutInv} & \gate{U} &\multigate{1}{\SQInOut} & \qw & \gate{U^\dagger} & \ghost{\SQInOutInv} & \qw & \qw & \qw & \qw \\
\lstick{\ket{D}} & \qw & \qw & \qw & \qw & \qw & \ghost{\SQInOut} & \gate{U} & \qw & \qw & \qw & \qw & \qw  & \qw \\
}
\]  \caption{Circuit for the modular exponentiation part of direct strategy for \( d=2 \), \( \log D = 4 \) and \( m=3 \) with 4 computational registers.}
  \label{fig:uncompute0}
\end{figure}
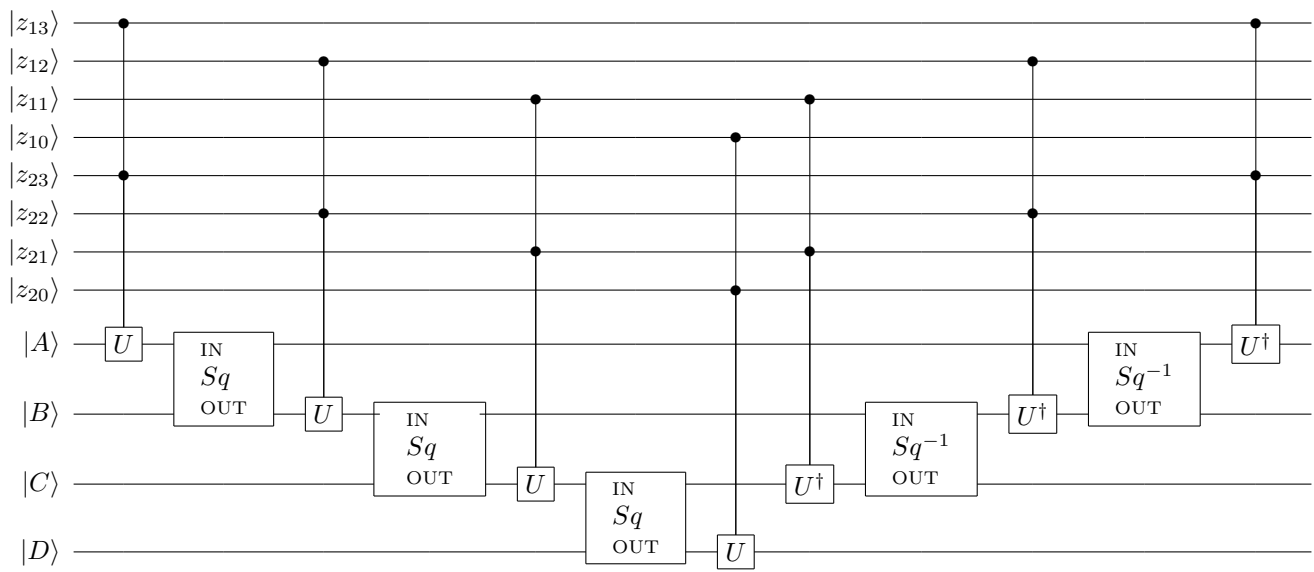

In this illustrative example, we also compare the resource usage of the intermediate-uncomputation scheme with the original scheme. In the original direct computation scheme, the circuit uses 5 squarings, 7 controlled-\( U \) operations (involving 5 large-integer multiplications and 7 small-integer multiplications), and 4 computational registers. In our new scheme, the circuit uses 5 squarings, 9 controlled-\( U \) operations (involving 5 large-integer multiplications and 9 small-integer multiplications), and 3 computational registers, which provides a preliminary illustration of the time--space trade-off.

This example is special because the uncomputation involves only the first computational register~\(\ket{A}\), which does not increase the number of squarings or large-integer multiplications. We further provide a side-by-side comparison for \( m \in \{3,5,7\} \); see Table~\ref{tab:space-time}, which more clearly shows the space-time trade-off.

\begin{table}[htbp]
\centering
\caption{Resource summary for different \( m \) and strategies (space = number of registers; time = counts of squarings, controlled-\( U \) operations, large-integer multiplications, and small-integer multiplications).}
\label{tab:space-time}
\begin{tabular}{@{} l l |S[table-format=1.0] |S[table-format=2.0] S[table-format=2.0] |S[table-format=2.0] S[table-format=2.0] @{}}
\toprule
\multicolumn{2}{c}{Setting} & \multicolumn{1}{c}{Space} & \multicolumn{4}{c}{Time} \\
\cmidrule(lr){1-2}\cmidrule(lr){3-3}\cmidrule(lr){4-7}
\( m \) & Strategy & {\# regs} & {Squares} & {Controlled-\( U \)} & {Large Mults} & {Small Mults} \\
\midrule
\multirow{2}{*}{\( 3 \)}
  & Simple & 3 & 5 & 9 & 5 & 9 \\
  & Direct & 4 & 5 & 7 & 5 & 7 \\
\addlinespace
\multirow{3}{*}{\( 5 \)}
  & Simple, \( k=2 \) & 4 & 11 & 15 & 11 & 15 \\
  & Simple, \( k=3 \) & 4 & 11 & 15 & 11 & 15 \\
  & Direct & 6 & 9 & 11 & 9 & 11 \\
\addlinespace
\multirow{4}{*}{\( 7 \)}
  & Simple, \( k=2 \) & 5 & 17 & 21 & 17 & 21 \\
  & Simple, \( k=4 \) & 5 & 17 & 21 & 17 & 21 \\
  & Binary & 4 & 19 & 27 & 19 & 27 \\
  & Direct & 8 & 13 & 15 & 13 & 15 \\
\bottomrule
\end{tabular}
\end{table}

Finally, we briefly discuss the case for practical use. If \( \log D = C \sqrt{n} \), the simple strategy requires approximately \( 2 \sqrt{C \sqrt{n}} \) computational registers. The binary recursion strategy requires approximately \( \log_2 (C \sqrt{n}) + 1 \) registers—both significantly optimized compared to Regev’s algorithm, which requires approximately \( C\sqrt{n} \) registers. For RSA-2048 (\( n = 2048 \)) with \( C = 2 \), the simple strategy uses approximately 19 registers, a substantial reduction from the approximately 90 registers in Regev’s algorithm, with a modest increase in time complexity. The binary recursion strategy further reduces space to approximately 8 registers, maintaining an asymptotically lower time complexity than Shor’s algorithm. Both strategies contribute to reduce the space cost of Regev’s algorithm, advancing its practical applicability.

\section{Performance of space-optimized Regev's algorithm}\label{Performance}

In this section, we present experimental demonstrations of Regev's algorithm and our algorithm. We focus on the minimal nontrivial factorization example of \( N = 35 \). We construct the basic building blocks of Regev’s algorithm and perform different levels of experimental simplification. These blocks are used for the following experimental demonstration and simulations.

We first show how the intermediate-uncomputation method affects the qubit count, circuit size, and performance in the presence of noise. Using an appropriate compilation technique, we construct circuits of 20 qubits (without intermediate-uncomputation) and 17 qubits (with intermediate-uncomputation), which allows us to perform noisy classical simulations of both versions and to compare their behavior. The simulated results show that our method not only reduces qubit consumption but also improves the success probability in the presence of noise, demonstrating the practical effectiveness of the proposed space-optimization method.

Finally, to execute Regev’s algorithm and obtain a factorization result on a real quantum computer, we adopt the classical precomputation method(described earlier and in Appendix~\ref{E}) and construct a 9-qubit circuit for factoring \( N = 35 \). We run the circuit on the IBM Brisbane quantum processor, apply lattice-based post-processing to recover the correct factors, and compare the experimental outputs with both ideal and noisy classical simulations. The results show that Regev’s algorithm can be implemented on currently available quantum hardware and can produce meaningful factoring outcomes.

\subsection{Construction of quantum circuit blocks for factoring \(N = 35\)}

In this section, we construct the quantum circuits for each building block using the simplest nontrivial example of factoring \(N = 35\) with small prime bases \(b_1 = 2\) and \(b_2 = 3\). We first introduce the selection of parameters, then describe the circuit construction and the experimental simplification process for Initialization, Modular Exponentiation, and Quantum Fourier Transform, showing the final result of circuits for our demonstration, with details provided in Appendices D, E, and F.

\begin{figure}[ht]
    \centering
    \includegraphics[width=0.5\textwidth]{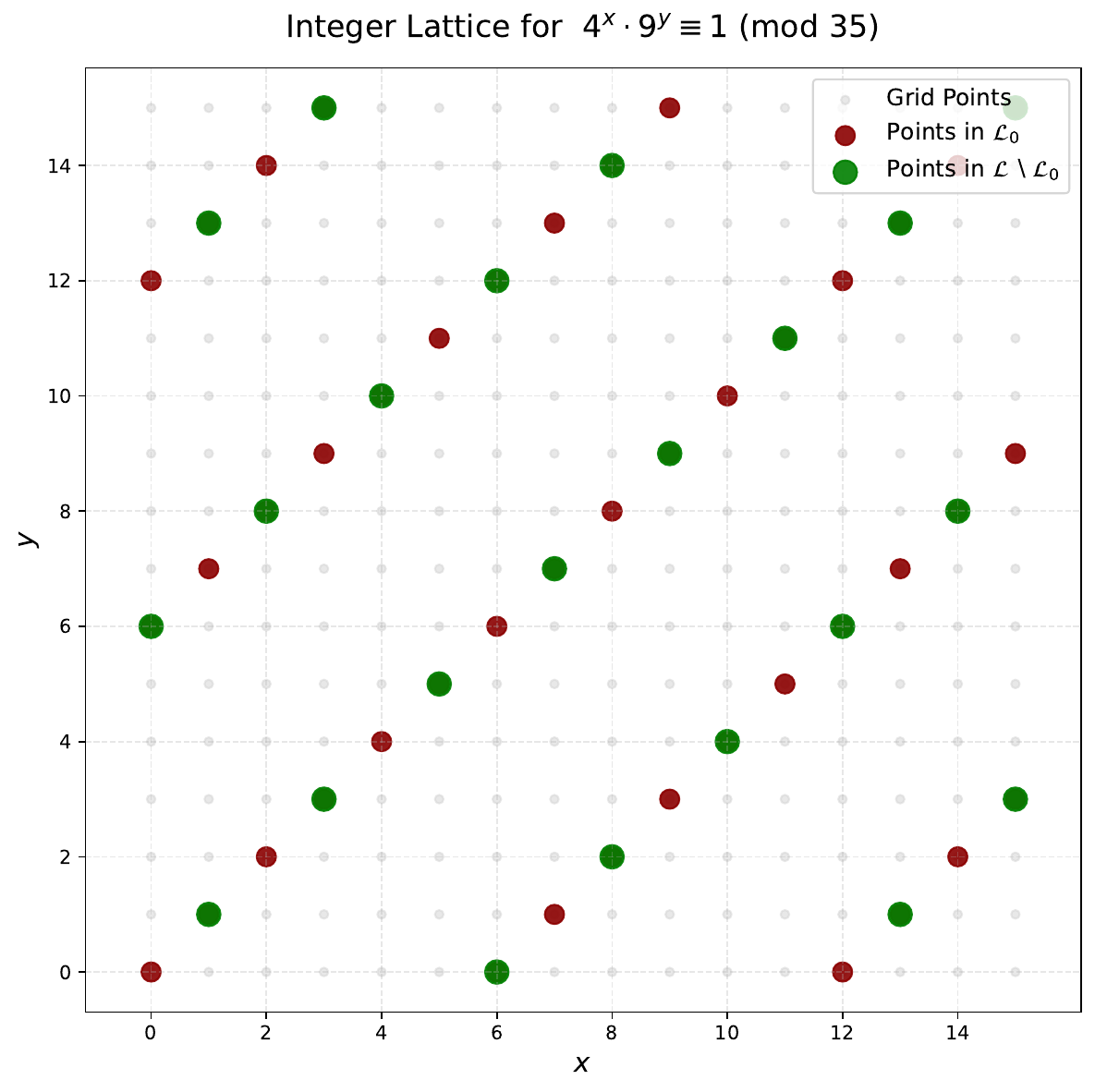}
    \caption{Two-dimensional integer lattice with points in \(\mathcal{L}\) that satisfy \(4^x \cdot 9^y \equiv 1 \pmod{35}\) (shown in dark red and green), and a subset of points that satisfy \(2^x \cdot 3^y \not\equiv \pm 1 \pmod{35}\) highlighted in green, for \(0 \leq x, y < 16\).}
    \label{fig:1}
\end{figure}

\subsubsection*{The choice of parameters in factoring \( N = 35 \):}
The parameters for this experimental demonstration are generally selected according to Regev’s work, where the selection rules are described in Appendix~\ref{D}.

The integer to be factored is \( N = 35 \), which has \( n = 6 \) bits. Therefore, each computational register requires \( n = 6 \) qubits\footnote{This can be reduced to 3 qubits, as detailed in Appendix~\ref{F}.}. We choose \( d = 2 \approx \sqrt{n} \), with the two small integers \( b_1 = 2 \) and \( b_2 = 3 \), and \( a_i = b_i^2 \), \( R > 2 \). The structure of the corresponding lattice \( \mathcal{L} \) is determined by the bases \( b_1 \) and \( b_2 \), and is shown in Figure~\ref{fig:1}.

Regarding the parameter \( D \), according to Regev's algorithm, a value of \( D > 4\sqrt{2} \) is sufficient. However, to demonstrate our method, as shown in the instances in Figures~\ref{fig:uncompute1} and \ref{fig:uncompute0}, \( D \) must be at least 16. Therefore, in the first experiment, designed to highlight the effect of our method and provide a clear comparison with the original construction, we choose \( D = 16 \); in the second experiment, which demonstrates Regev's algorithm using minimal resources, we select the smallest admissible value, \( D = 8 \). In the following sections, we use \( D = 8 \) to illustrate the basic building blocks. Other detailed parameter selections can be found in Appendix~\ref{D}.

\subsubsection*{Quantum state initialization:}

For each block of control registers, as mentioned earlier, we only need to consider the effect of the \( O(\log d) \) most significant qubits out of the total \( \log D \) qubits. For this small-scale example, we can arbitrarily choose the first 2, 1, or even 0 qubits as the most significant. Using the technique in~\cite{Long01}, we obtain the four circuits shown in Figure~\ref{fig:initialization}.

\begin{figure}[ht]
\centering
\includegraphics[width=0.6\textwidth]{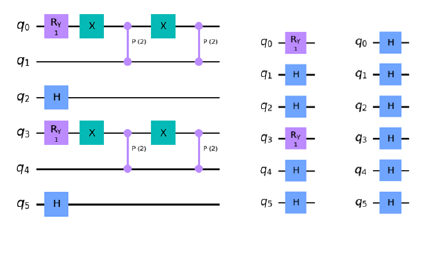}
\caption{Circuit implementation of three different quantum state initialization methods.The figure shows initialization circuits for selecting the most significant \( O(\log d) \) qubits from \( \log D \), corresponding to selecting 2, 1, or 0 qubits. The left diagram selects 2 qubits, the middle selects 1, and the right selects 0. The approach shown is based on the technique in~\cite{Long01}.}
\label{fig:initialization}
\end{figure}

Indeed, we can directly use a uniform superposition state as the initial input, which also yields correct results on a small scale, significantly simplifying the circuit.

\subsubsection*{Modular exponentiation and compiled version of circuit:}

To compute the value \( \prod_i a_i^{z_i + D/2} \pmod{N} \) in the second register, since \( z_i + D/2 \) is a positive number, we can directly treat \( z_i + D/2 \) as the binary representation of the \( i \)-th control register, and denote \( z_{ij} \) as the \( j \)-th bit of the \( i \)-th part of the control register \( z_i \) (which introduces differences in preparation; however, since we prepare the initial state as a uniform superposition, this does not introduce any changes).
    
\[
\ket{z} \ket{0} \to \ket{z} \ket{\prod_i a_i^{z_i} \pmod{N}}.
\]

As mentioned earlier, two distinct architectures can be employed to implement modular exponentiation, with circuit decomposition, ultimately requiring the following quantum operations:

\begin{itemize}
    \item Square-and-multiply method.  
    This is done by sequentially applying the controlled-\( U \) operator conditioned on \( z_{ij} \), to compute \( \prod_i a_i^{z_{ij}} \mod{N} \) together with the square operator \( U_{sq} \). A diagram illustrating this can be found in Figure~\ref{fig:uncompute0}. Regev's algorithm relies on this scheme to achieve an asymptotic time complexity advantage over Shor's algorithm.

    \item Precomputation method.  
    This is done by applying the controlled-\( U^{2^j} \) operator sequence conditioned on \( z_{ij} \), to compute \( \prod_i a_i^{z_{ij} 2^j} \mod{N} \). A diagram illustrating this can be found in Figure~\ref{fig:circuit0}. This scheme is more suitable for small-scale experimental demonstrations, but it does not provide an asymptotic complexity advantage.
\end{itemize}

Both modular exponentiation schemes ultimately reduce the modular exponentiation operation to computing several modular multiplications controlled-\( U \) (and modular squarings \( U_{sq} \)). The modular multiplications controlled-\( U \) compute \( \prod_i a_i^{z_{ij}} \mod{N} \) and can be decomposed into the controlled multipliers \( U_4 \) and \( U_9 \) in our example. The modular multiplications controlled-\( U^{2^j} \) compute \( \prod_i a_i^{z_{ij} 2^j} \mod{N} \), and can be decomposed into the controlled multipliers \( U_4^{2^j} \) and \( U_9^{2^j} \). Here, \( U_c \) denotes modular multiplication by \( c \bmod N \) on the computation register. Finally, \( U_{sq} \) represents the squaring operation on the computation register. In our work, we use the square-and-multiply architecture to demonstrate the effect of our uncomputation-based method, and the precomputation architecture to implement Regev's algorithm on real quantum hardware.

Regarding the choice of multiplier, previous experimental studies often use the ``compiled'' multiplier for demonstrating factorization with very small numbers, such as \( N = 15, 21 \). This method requires fewer gates and ancilla qubits than a general multiplier, as described in~\cite{Lanyon07}~\cite{Beckman96} and summarized in~\cite{Monz16}. We adopt these techniques from previous works and apply several additional simplification techniques (see Appendix~\ref{E}). The correctness of the multiplier can be easily verified, such as for \( 2^j \pmod{15} \), as shown in Figure~\ref{fig:multiplication_15}~\cite{Monz16}. We adopt these techniques from previous works and apply several additional simplification techniques. Full circuit-level results and further details are presented in Appendix~\ref{F}. A compiled multiplier is shown in Figure~\ref{fig:multiplier_35}; the squaring operation requires ancilla qubits even in the compiled version, as illustrated in Figure~\ref{fig:square_35}.

The above units are sufficient to construct the modular exponentiation, but to further simplify the quantum circuit and reduce the number of qubits, we can adopt a more simplified approach. In practice, we can further simplify the quantum circuit in two ways: re-compilation (which reduces the number of qubits required for a single computational register from 6 to 3)and a simplified multiplier from the initial state (which simplifies the first few multiplier). These techniques are detailed in Appendix~\ref{F}, leading to varying degrees of simplification from Figure~\ref{fig:multi_sim} to Figure~\ref{fig:square_sim}, which can significantly reduce the number of qubits and circuit size.

\subsubsection*{Quantum Fourier transform:}

Apply the \( d \)-dimensional quantum Fourier transform (over \( \mathbb{Z}_D^d \)) to the \( \ket{z} \) register, which is equivalent to applying a quantum Fourier transform over \( \mathbb{Z}_D \) to each dimension separately, which is easily implemented as shown in Figure~\ref{fig:QFT}.

\begin{figure}[ht]
\centering
\includegraphics[width=0.5\textwidth]{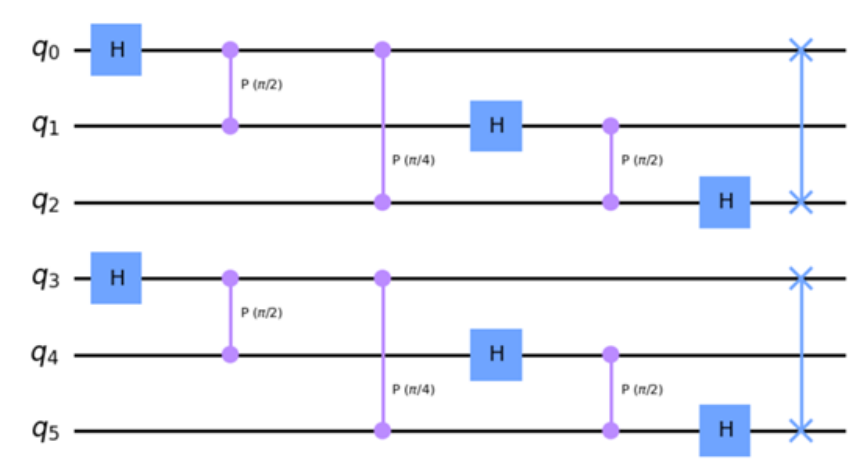}
\caption{Circuit implementation of the quantum Fourier transform over \( \mathbb{Z}_8^2 \).}
\label{fig:QFT}
\end{figure}

Given all the modules designed above, we can now construct the specific quantum circuits for our method and Regev’s algorithm to demonstrate the proof-of-concept through simulations and experiments.


\subsection{Simulation and experimental results}

In this section, we construct quantum circuits using the basic building blocks described above and present the corresponding simulation and experimental results. We first build two comparative circuits and report noisy-simulation results obtained on a quantum simulator, which demonstrate the performance of our space-optimized algorithm and the effectiveness of the proposed space-optimization method. We then construct a further simplified circuit suitable for current quantum hardware and present experimental data demonstrating Regev's factoring approach on a real quantum device. Finally, we run classical post-processing routines to complete the factorization and verify correctness.

\subsubsection{The simulation result for intermediate-uncomputation method}


Firstly, we construct the two circuits shown in Figures~\ref{fig:uncompute1} and~\ref{fig:uncompute0} for the case where \( \log D = 4 \), following Regev’s square-and-multiply architecture to demonstrate the practical effect of our theoretical optimization.

For the simplification, we use a uniform superposition initialization, along with the multipliers shown in Figures~\ref{fig:multi_sim} to~\ref{fig:square_sim}, to construct the controlled multipliers \( U_4 \), \( U_9 \), and the out-of-place squaring operation. Each register requires 3 qubits, resulting in a total of 17 qubits for Figure~\ref{fig:uncompute1} and 20 qubits for Figure~\ref{fig:uncompute0}. The final circuits are shown in the figures in Appendix~\ref{G}. Subsequently, we run these circuits using Qiskit Aer to obtain simulated results, including both noiseless simulations and simulations under the FakeBrisbane noise model. The ideal results without noise are presented in Figure~\ref{fig:result0}, and the simulation results under noisy conditions are shown in Figures~\ref{fig:result17} and~\ref{fig:result20}.

\begin{figure}[htbp]
\centering
\includegraphics[width=0.7\textwidth, height=0.4\textheight, keepaspectratio]{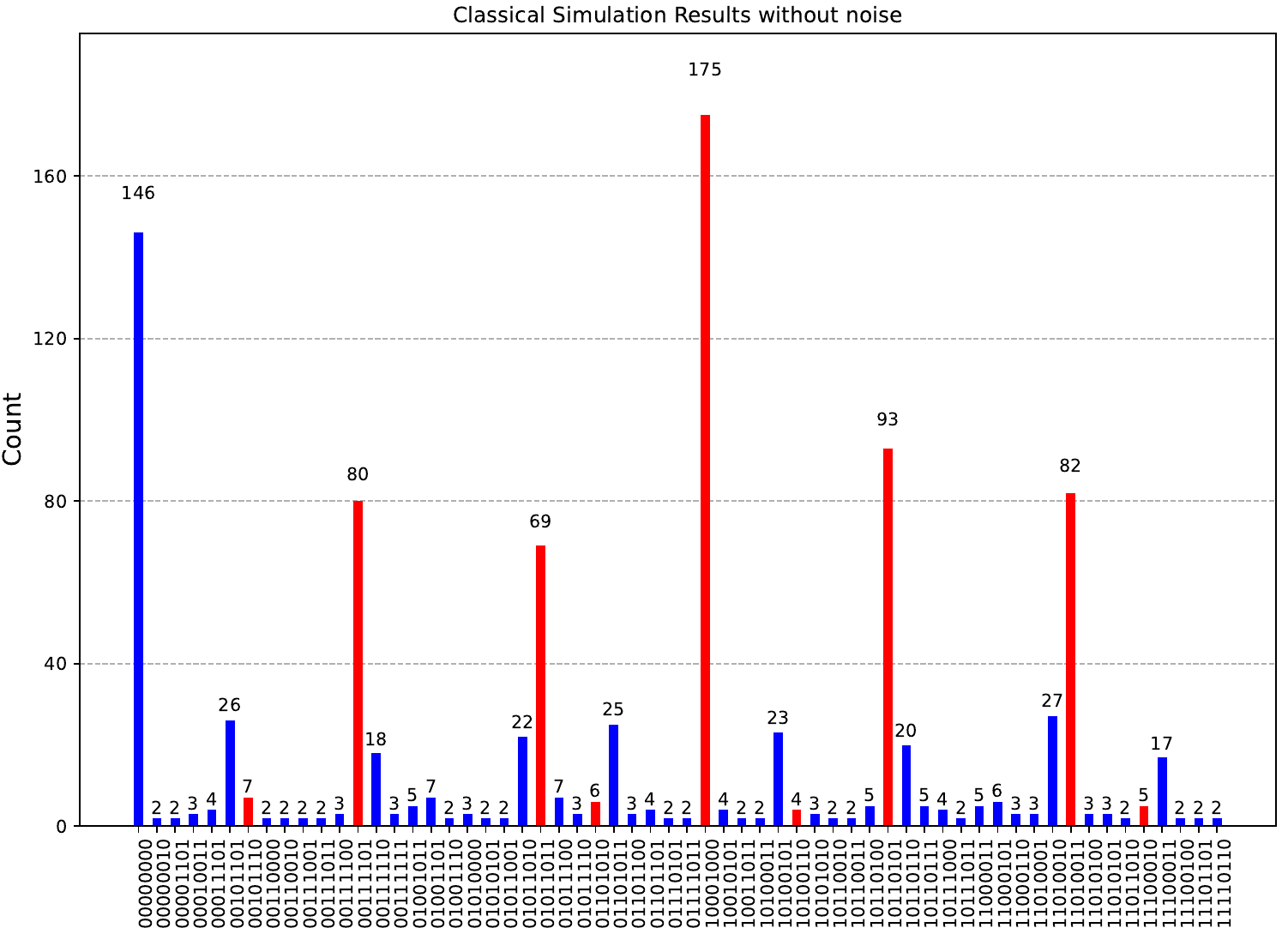}
\caption{Classical simulation results without noise. The red bars in the figure represent the states that are expected to be measured, located near the dual lattice elements.}
\label{fig:result0}
\end{figure}

Since these two circuits are equivalent, the results are identical in noiseless conditions. The results of the simulation without noise in Figure~\ref{fig:result0} indicate that the measurement outcomes with the highest probabilities correspond to the peaks, which are:

\begin{align*}
10001000 &: (1/2, 1/2), \\
00000000 &: (0, 0), \\
10110101 &: (11/16, 5/16), \\
11010011 &: (13/16, 3/16), \\
00111101 &: (13/16, 3/16), \\
01011011 &: (5/16, 11/16).
\end{align*}

These states are the states in the dual lattice \( L^* \), which are the desired states.

The simulation results with noise in Figure~\ref{fig:simulation_noise} show that, under the same noise model, although our method requires slightly more gates, the probability of measuring the target state corresponding to the dual lattice vector is significantly higher. Specifically, our method achieved 415 successful counts out of 4096, compared to 330 successful counts out of 4096 for the direct method. Since the simulator mainly models noise in quantum gates, the impact of idle qubits is more pronounced in real quantum devices. In practical quantum computers that can support the computation of our circuits, the difference between the two methods may be even more significant. These results indicate that our method not only reduces the required number of qubits but also has practical significance in reducing the impact of quantum noise in real quantum computers.

\begin{figure}[htbp]
\centering
\begin{subfigure}{0.48\linewidth}
\centering
\includegraphics[width=\linewidth,height=0.6\textheight,keepaspectratio]{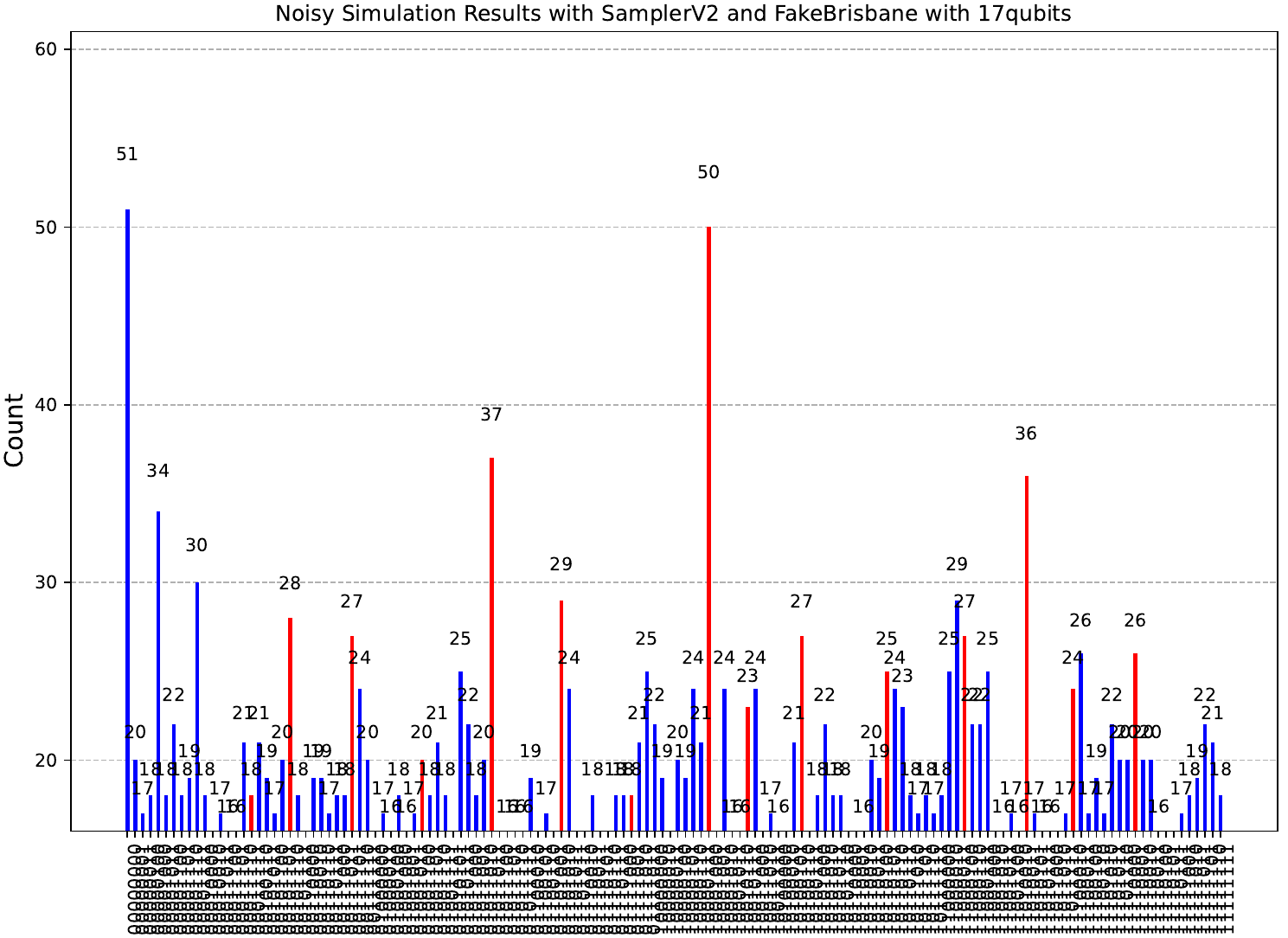}
\caption{}
\label{fig:result17}
\end{subfigure}
\hfill
\begin{subfigure}{0.48\linewidth}
\centering
\includegraphics[width=\linewidth,height=0.6\textheight,keepaspectratio]{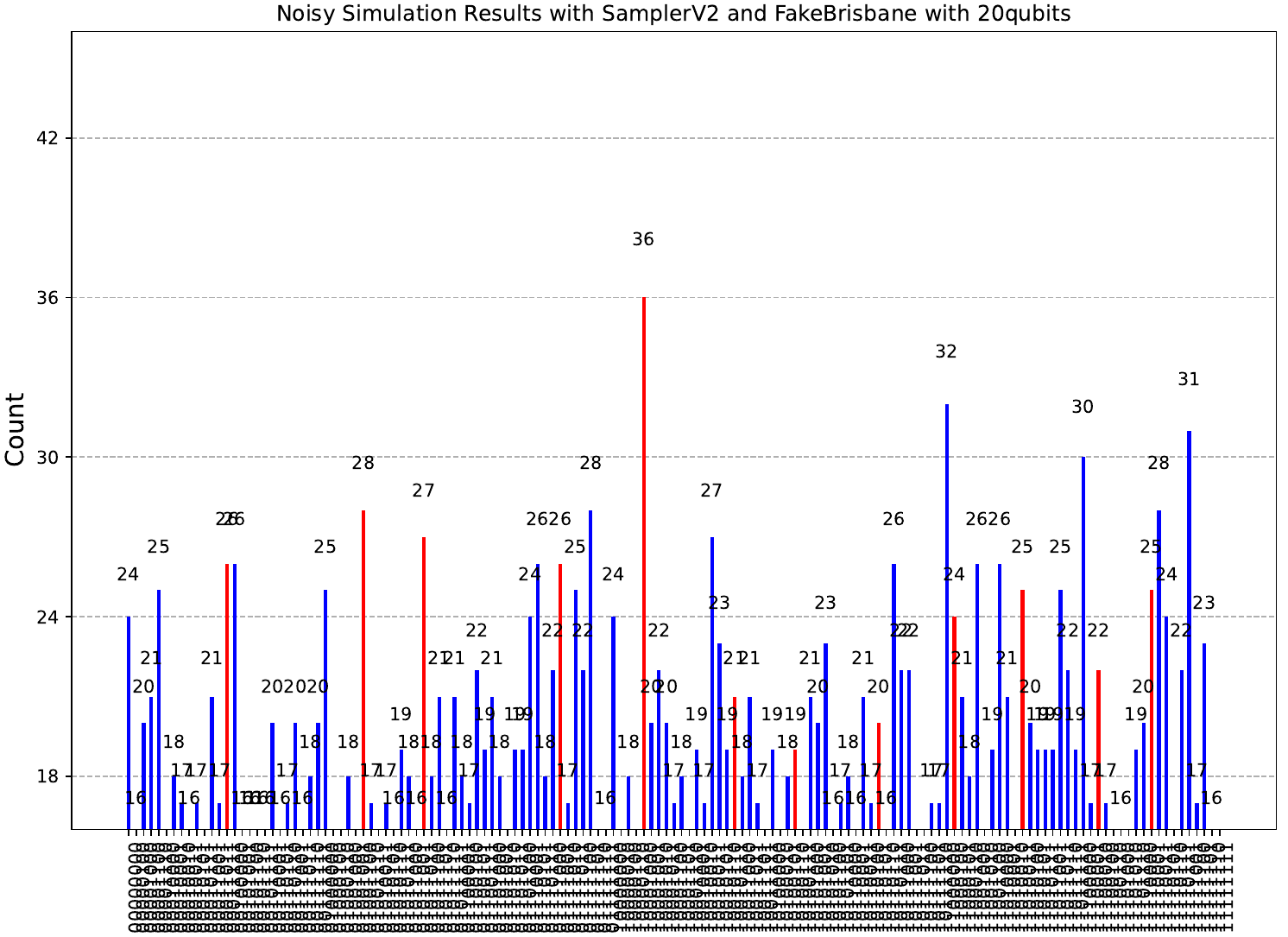}
\caption{}
\label{fig:result20}
\end{subfigure}
\caption{Simulation results with noise. (a) Simulation of the circuit from Figure~\ref{fig:uncompute1} using our intermediate uncomputation method. (b) Simulation of the circuit from Figure~\ref{fig:uncompute0} in the Regev's original method. Both simulations were performed with Qiskit Aer using optimization level 3 and a noise model extracted from FakeBrisbane, representing typical noise levels in current quantum processors. In presenting the results, states with fewer measurements were filtered.}
\label{fig:simulation_noise}
\end{figure}


\subsubsection{The simulation and experimental result for the Regev's precompute demonstration}\label{experiment}


Next, we choose \( \log D = 3 \) and conduct experiments using the classical precomputation scheme, which uses fewer quantum resources and is typically more practical for small-scale demonstration experiments. This approach allows us to classically compute the constants \( a_i^{2^j} \) and compile them into the multiplier circuit, avoiding the need for squaring operations, and facilitating a proof-of-concept demonstration on real quantum hardware.

We apply the controlled-\( U \) operator sequence \( \{ U_4^4, U_9^4, U_4^2, U_9^2, U_4, U_9 \} \) to compute \( \prod_i a_i^{z_i} \mod{N} \) on the second register, conditioned on the first register. The multipliers from Figure~\ref{fig:multi_sim} are used to construct the controlled multipliers \( U_4^{2^j} \) and \( U_9^{2^j} \). We then construct the 9-qubit precomputation-based circuit and execute it both in Qiskit Aer and on the IBM Quantum Brisbane processor to obtain simulated and experimental results, and perform classical post-processing to verify the correctness of the algorithm. The overall quantum circuit for this example is fully represented in Figure~\ref{fig:circuit0}. We briefly introduce the simulation results, the quantum hardware used, and compare the experimental results with the noisy simulation, running post-processing programs to complete the full factorization process.


We first use Qiskit Aer to simulate the circuit results, as shown in Figure~\ref{fig:simulation_results}. The results of the simulation indicates that the measurement outcomes with the highest probabilities correspond to the red peaks in Figure~\ref{fig:result1}, which are:

\begin{align*}
100100 &: (1/2, 1/2), \\
000000 &: (0, 0), \\
001111 &: (1/8, 7/8), \\
111001 &: (7/8, 1/8), \\
101011 &: (5/8, 3/8), \\
011101 &: (3/8, 5/8).\\
\end{align*}

These states are the desired states in the dual lattice \( L^* \).

\begin{figure}[htbp]
\centering
\begin{subfigure}{0.48\linewidth}
\centering
\includegraphics[width=\linewidth,height=0.6\textheight,keepaspectratio]{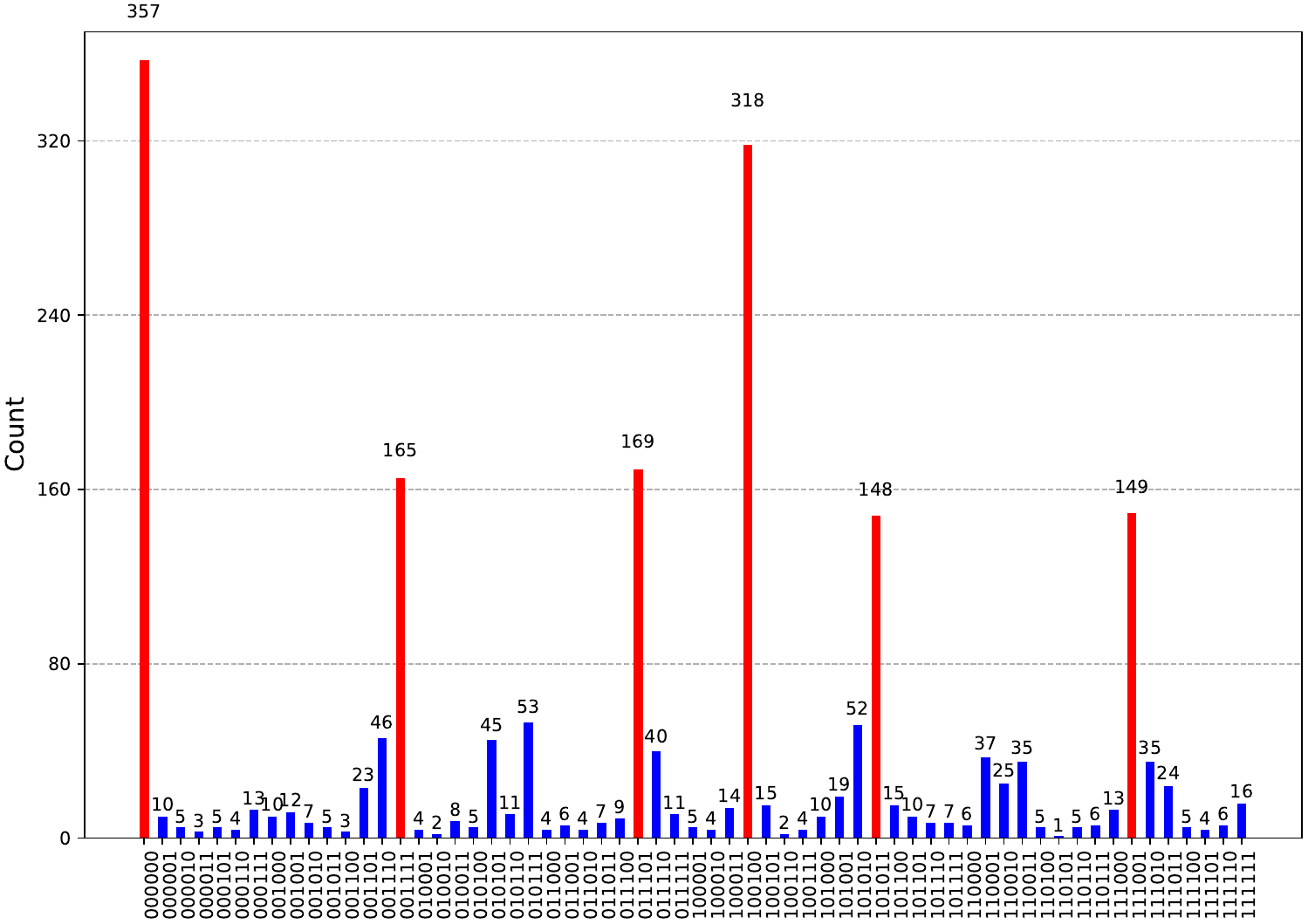}
\caption{}
\label{fig:result1}
\end{subfigure}
\hfill
\begin{subfigure}{0.48\linewidth}
\centering
\includegraphics[width=\linewidth,height=0.6\textheight,keepaspectratio]{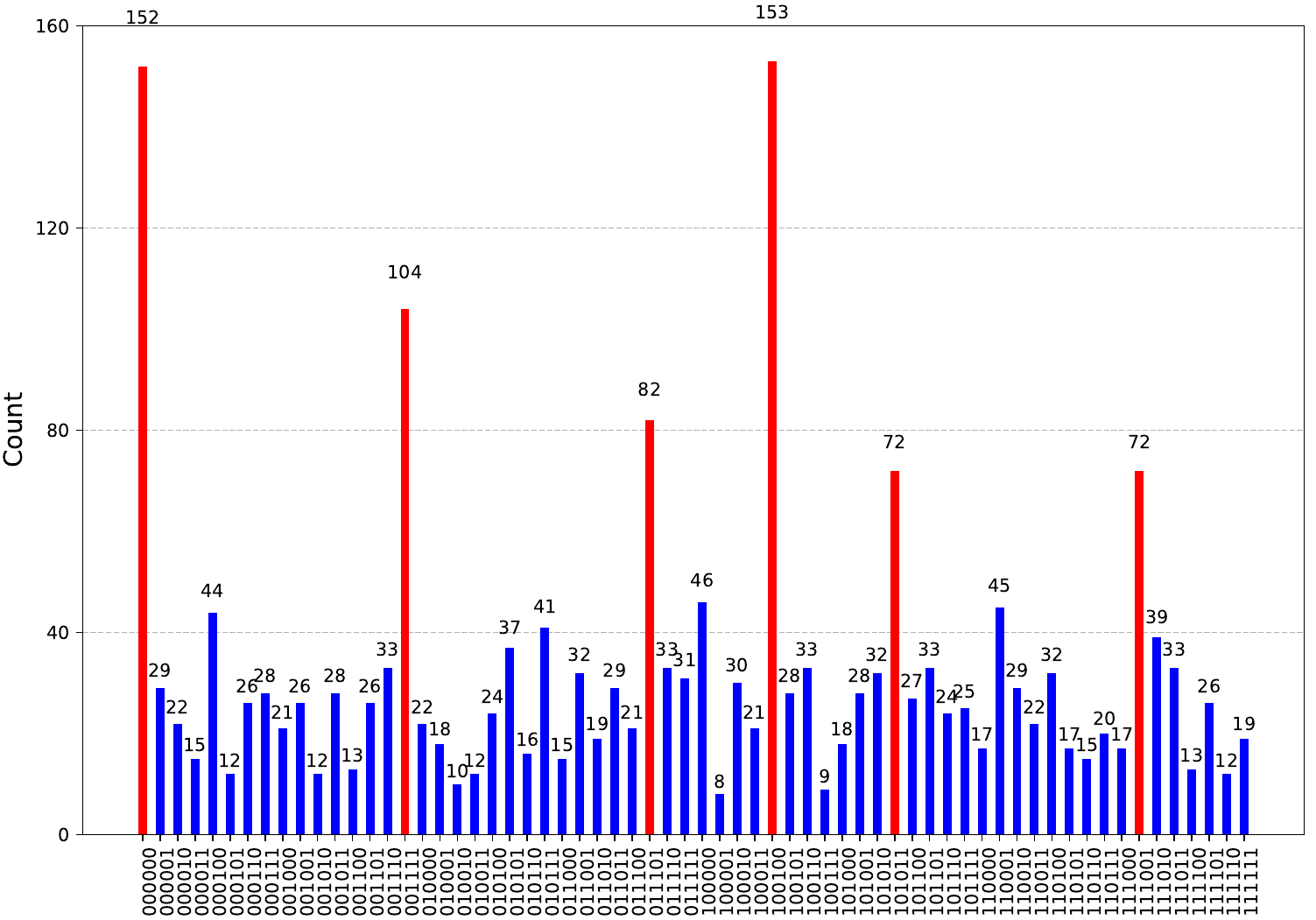}
\caption{}
\label{fig:result2}
\end{subfigure}
\caption{Classical simulation results with a 9-qubit precompute circuit. (a) Simulation without noise. (b) Simulation with noise from Fake Brisbane. The red bars represent the target states.}
\label{fig:simulation_results}
\end{figure}


To demonstrate the algorithm’s performance on a real quantum computer, we executed the algorithm on the IBM Brisbane processor. The IIBM Brisbane processor is a 127-qubit superconducting device of the Eagle family, to implement experiments for factoring the number \( N = 35 \). The qubits are arranged in a heavy-hex lattice, with their connectivity illustrated schematically in Figure~\ref{fig:brisbane_coupling}. Our circuit uses a connected subgraph of 9 physical qubits to host the logical registers, chosen according to the device connectivity and calibration data at the time of the experiments.

\begin{figure}[h]
\centering
\includegraphics[width=2.5in, height=2.5in]{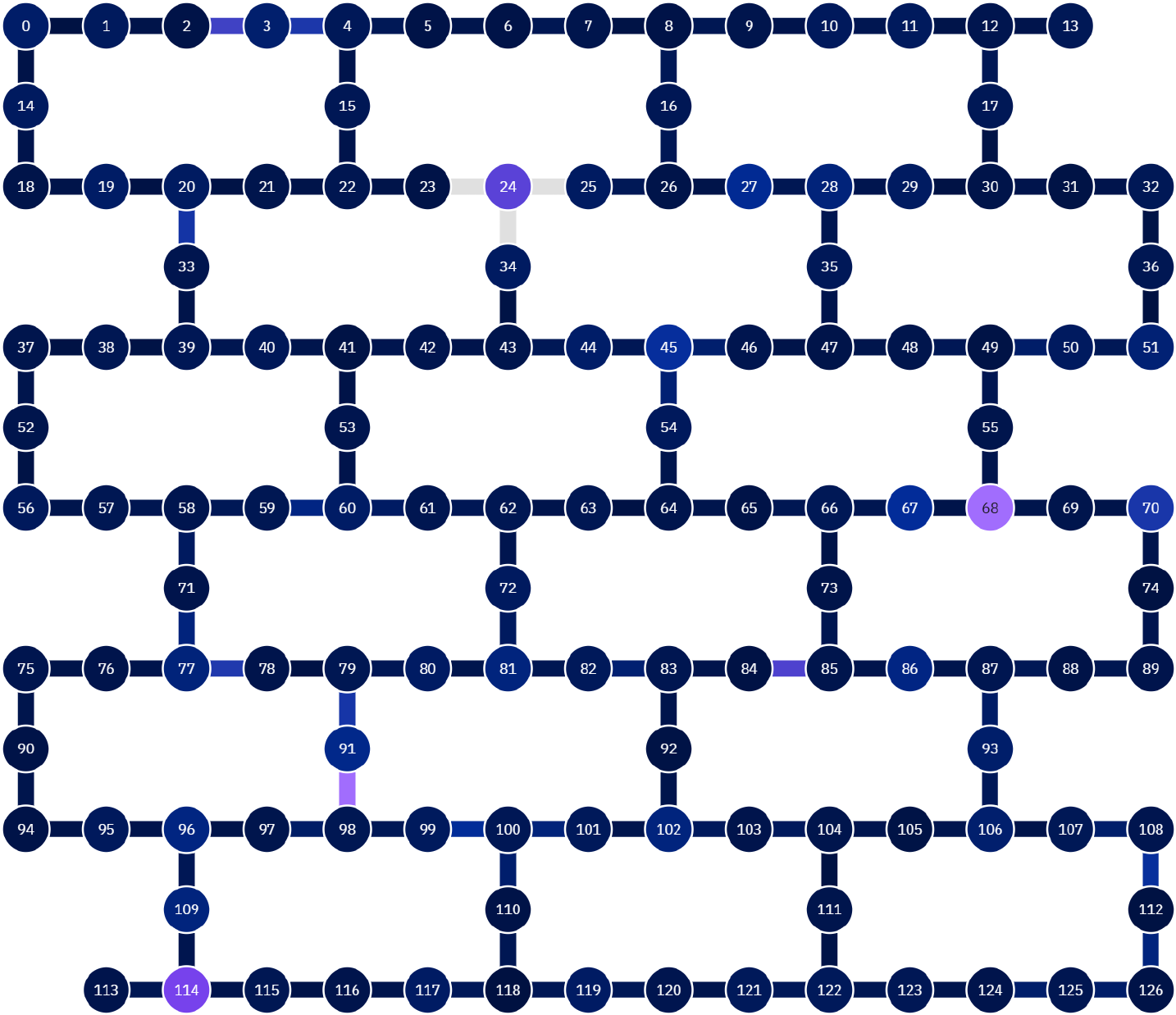}
\caption{Coupling map of the IBM Quantum Brisbane processor (Eagle r3, 127 qubits).}
\label{fig:brisbane_coupling}
\end{figure}

Representative calibration data for the IBM Quantum Brisbane processor, reported in recent hardware-characterization work on Eagle-r3 devices, indicate median coherence times on the order of \(T_1 \approx 220\text{--}230~\mu\mathrm{s}\) and \(T_2 \approx 140\text{--}150~\mu\mathrm{s}\), single-qubit \(\mathrm{SX}\)-gate error rates of order \(2\text{--}3\times 10^{-4}\), two-qubit ECR error rates around \(8\times 10^{-3}\), and readout assignment errors at the \(1\text{--}2\%\) level~\cite{AbuGhanem2025IBMQC}. A more recent benchmarking study focused specifically on IBM Brisbane reports performance figures consistent with this parameter range~\cite{Selvam2025Brisbane}. Our hardware experiments, performed on IBM Brisbane in June~2025, were therefore carried out in a noise regime well represented by these reported calibration values.


We compared the simulation results using the noise model FakeBrisbane (Figure~\ref{fig:result3}) with the actual computational results from the IBM Brisbane quantum computer (Figure~\ref{fig:result4}). It is observed that the noise level of the real quantum computer is higher than that of its noise model. Under actual noise conditions, the probability of measuring the target state is significantly lower than the theoretical expectation, but it remains among the higher-probability outcomes. This demonstrates the effectiveness of the algorithm on current noisy quantum computers. 

If additional techniques are applied (see Potential Improvements for Future Work), the required number of qubits and two-qubit gates could be further reduced, potentially minimizing the impact of noise. Moreover, using more advanced quantum hardware could yield even better results.

\begin{figure}[htbp]
\centering
\begin{subfigure}{0.48\textwidth}
\centering
\includegraphics[width=\linewidth]{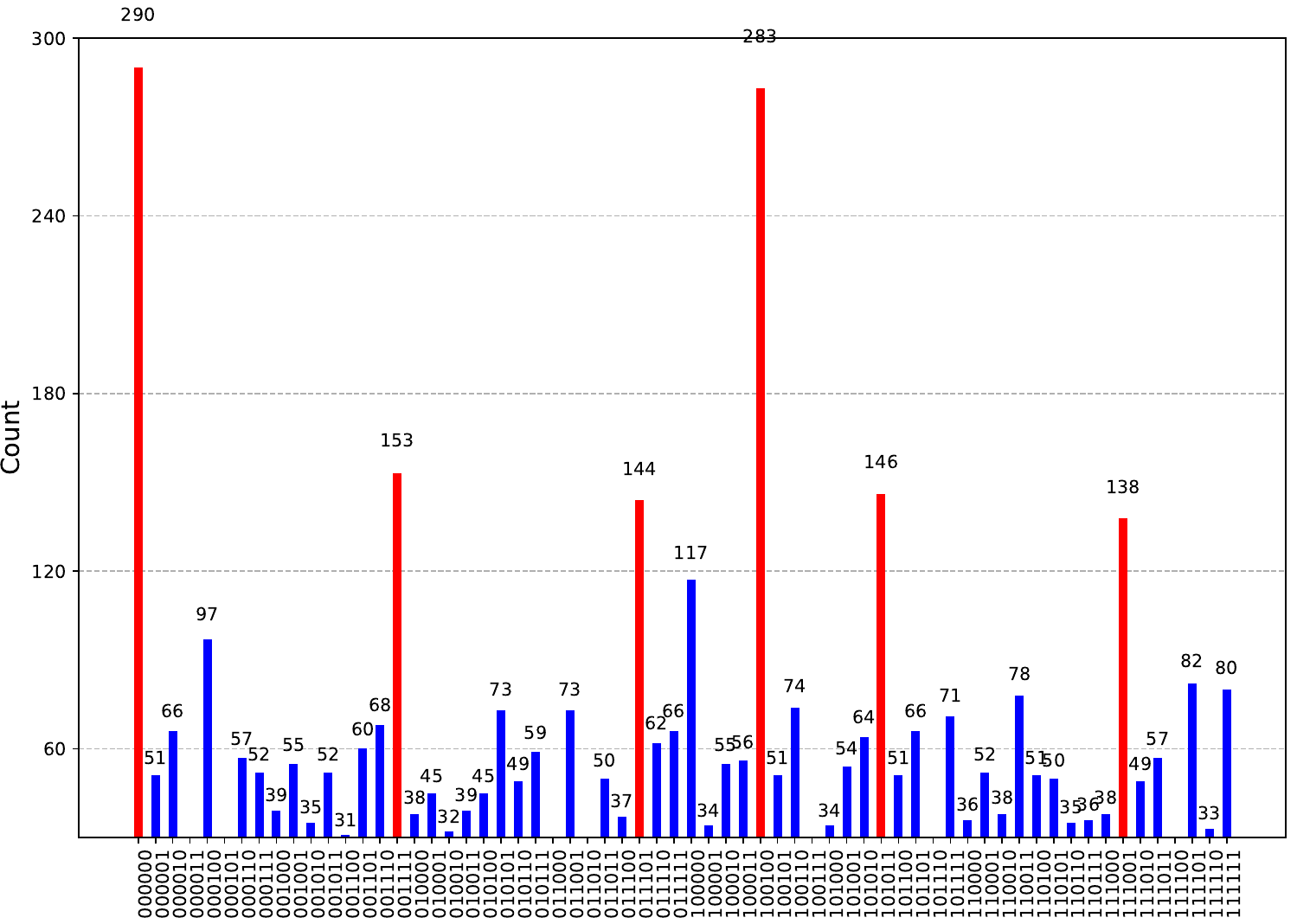}
\caption{}
\label{fig:result3}
\end{subfigure}
\hfill
\begin{subfigure}{0.48\textwidth}
\centering
\includegraphics[width=\linewidth]{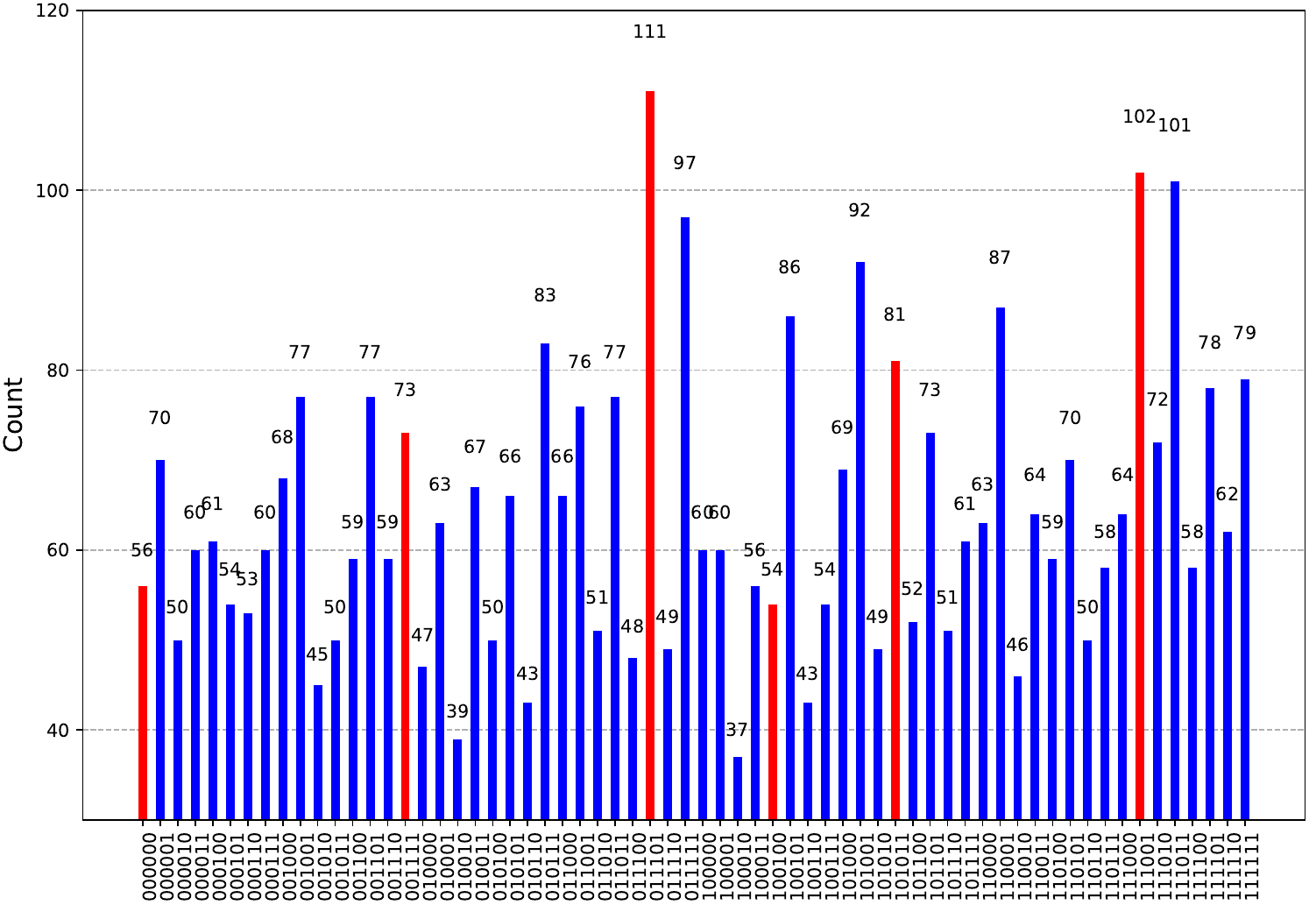}
\caption{}
\label{fig:result4}
\end{subfigure}
\caption{Comparison of simulation and experimental results with 9-qubit precompute circuit. (a) Noisy simulation results obtained with the noise model extracted from FakeBrisbane. (b) Experimental results obtained on the real IBM Brisbane quantum processor.}
\label{fig:comparison_results}
\end{figure}

\subsubsection{Post-processing}

Finally, we perform post-processing on the lattice based on the experimental results and complete the correct factorization. The data are sourced from subsection~\ref{experiment}.

\begin{itemize}
\item \textbf{Generation of lattice \(\mathcal{L}'\) and LLL reduction}:  

By inputting the measurement results into the generating matrix for post-processing, we obtain the lattice \(\mathcal{L}'\), on which we perform LLL reduction to complete the factorization process.

For example, consider the four measured vectors \(\{ (1/2, 1/2), (1/8, 7/8), (3/8, 5/8), (0, 0) \}\). Setting \( S = 8 \), the generating matrix is:
\[
\begin{bmatrix}
1 & 0 & 0 & 0 & 0 & 0 \\
0 & 1 & 0 & 0 & 0 & 0 \\
4 & 4 & 8 & 0 & 0 & 0 \\
1 & 7 & 0 & 8 & 0 & 0 \\
3 & 5 & 0 & 0 & 8 & 0 \\
0 & 0 & 0 & 0 & 0 & 8
\end{bmatrix}
\]
After applying the LLL algorithm, the basis matrix becomes:
\[
\begin{bmatrix}
-1 & 1 & 1 & -1 & -3 & 0 \\
-1 & -1 & 0 & 0 & 3 & 0 \\
0 & 0 & 4 & 4 & 0 & 0 \\
0 & 2 & 1 & -1 & 2 & 0 \\
0 & -2 & 3 & -3 & -2 & 0 \\
0 & 0 & 0 & 0 & 0 & 8
\end{bmatrix}
\]
Taking the first element, corresponding to \((-1, -1)\), it is verified that \( 4^{-1} \cdot 9^{-1} \equiv 1 \pmod{35} \), which is indeed an element of \( L \). Furthermore, since \( 2^{-1} \cdot 3^{-1} \equiv 18 \cdot 12 \pmod{35} = 6 \neq 1 \), it is not an element of \( L_0 \).

\item \textbf{Final result}:  

Computing the greatest common divisor \( \gcd(2^{-1} \cdot 3^{-1}, 35) \) yields 5 and 7 as the nontrivial factors of 35, thus successfully completing the factorization.

\end{itemize}

At this point, we have completed the entire factorization process, thereby completing a proof-of-principle demonstration and validating the correctness and effectiveness of the quantum algorithm.

\section{Discussion}\label{Disc}

In this work, we have proposed and analyzed a intermediate-uncomputation framework for Regev's quantum factoring algorithm. By strategically uncomputing intermediate results, we reduce the number of required computational registers while maintaining a polynomial time factor that is more efficient than Shor's algorithm. Theoretical scaling laws, and simulation results for small-scale instances collectively confirm the effectiveness of this approach. And a proof-of-principle demonstration of Regev's algorithm is excuted on real quantum hardware.

The demonstrated space--time trade-off offers a practical pathway for implementing Regev's algorithm on current and future quantum processors. In fact, our uncomputation-based space-reduction method is broadly applicable and can be extended to other quantum algorithms, providing a new pathway toward practical quantum computation.

There are several directions for further improvement. On the theoretical side, it is worth emphasizing that we primarily measured time complexity in terms of circuit size (i.e., large-integer squarings and multiplications). If we focus on circuit depth, the simple intermediate-uncomputation strategy shows a depth overhead closer to the direct strategy, as many operations within a block can be parallelized. The k-ary intermediate-uncomputation strategy also significantly benefits from parallelization to improve the time factor, with the specific results requiring more detailed analysis. However, the asymptotic lower bounds on space complexity remain unaffected. A more refined analysis that considers hardware connectivity and parallelization would be a valuable extension.

We note that recent work by Kahanamoku-Meyer et al.~\cite{Kah2025}, based on parallel spooky pebble games, has been posted on arXiv. While the idea of parallelizing quantum circuits is natural, spooky pebble games rely on intermediate measurements (and gate operations conditioned on measurement outcomes), which seems challenging to implement on current quantum hardware. Our work is independent of theirs, as we focus on optimizing Regev’s algorithm in terms of both space complexity and practical small-scale experimental implementation, without relying on intermediate measurements and can be implemented on NISQ devices. If intermediate measurements become feasible at low cost in future quantum computers, the spooky pebble game holds significant promise, as it could achieve reduced time complexity and potentially surpass the theoretical space lower bound in our setting. This could become a key research direction for both theoretical and experimental studies in the future.

On the implementation side, besides the experimental optimization techniques discussed in this paper, some unadopted schemes exist to potentially reduce the number of qubits and circuit size. For instance, the ``semi-classical Fourier transform'' could be employed to reduce the number of control qubits. However, this approach requires adjusting the parameters of quantum rotation gates after measurement, which also appears to be unsupported on IBM quantum machines now. When allowing Toffoli gate decomposition with phase deviations, we could further reduce the circuit size. These simplification schemes may be explored in future work.

\clearpage

\bibliographystyle{unsrt} 
\bibliography{References}

\clearpage

\appendix
\clearpage

\section{Appendix A: Overview of Shor's algorithm }\label{A}

Shor's algorithm can efficiently solve the following integer factoring problem: Given an \( L \)-bit odd composite integer \( N = pq \), determine its nontrivial factors \( p \) and \( q \).
This is achieved by reducing the factorization problem to an order-finding problem, a task that can be efficiently performed by a quantum subroutine with high probability.

The core idea is as follows: randomly choose a positive integer \( a < N \) that is co-prime to \( N \). If we can find the order \( r \) of \( a \) modulo \( N \)—that is, the smallest positive integer such that \( a^r \equiv 1 \pmod{N} \)—then, according to a mathematical theorem, with probability greater than one-half, \( r \) is even and \( a^{r/2} \not\equiv \pm 1 \pmod{N} \) \cite{Shor94}. Under this condition, the identity \( (a^{r/2} + 1)(a^{r/2} - 1) \equiv 0 \pmod{N} \) holds, and since \( 1 < a^{r/2} \pm 1 < N \), computing \( \gcd(a^{r/2} \pm 1, N) \) is guaranteed to yield a nontrivial factor of \( N \).

\subsubsection*{Order-finding Subroutine}

\begin{enumerate}
    \item \textbf{Quantum State Initialization}: Prepare two quantum registers. The first register, with \( n = \lceil 2 \log_2 N \rceil = 2L + 1 \) qubits, is initialized to a uniform superposition over all states \( \ket{x} \) for \( x = 0, 1, \dots, 2^n - 1 \) using Hadamard gates. The second register, with \( m = \lceil \log_2 N \rceil \) qubits, is initialized to \( \ket{0} \) to hold the result of modular exponentiation in the next step:
    \[
    \ket{\psi_1} = \frac{1}{\sqrt{2^n}} \sum_{x=0}^{2^n-1} \ket{x} \ket{0}.
    \]

    \item \textbf{Modular Exponentiation}: Apply the unitary operator \( U_a \) to compute \( a^x \pmod{N} \) in the second register, conditioned on the first register:
    \[
    \ket{x} \ket{0} \to \ket{x} \ket{a^x \pmod{N}}.
    \]
    This yields the state:
    \[
    \ket{\psi_2} = \frac{1}{\sqrt{2^n}} \sum_{x=0}^{2^n-1} \ket{x} \ket{a^x \pmod{N}}.
    \]

    \item \textbf{Inverse Quantum Fourier Transform}: Apply the inverse Quantum Fourier Transform (QFT\(^{-1}\)) to the first register to extract periodicity information:
    \[
    \ket{x} \to \frac{1}{\sqrt{2^n}} \sum_{k=0}^{2^n-1} e^{-2\pi i x k / 2^n} \ket{k}.
    \]
    This transforms the state to reveal peaks at multiples of \( 2^n / r \).

    \item \textbf{Measurement}: Measure the first register to obtain a value \( k \), which, with high probability, is approximately \( \frac{2^n}{r} \cdot j \) for some integer \( j \). This measurement provides a value for further processing.

    \item \textbf{Continued Fraction Expansion}: Use the continued fraction algorithm on the measured value \( \frac{k}{2^n} \) to approximate a fraction \( \frac{j}{r} \). If \( \gcd(j, r) = 1 \), then \( r \) is the order. Repeat the process if necessary to confirm \( r \).
\end{enumerate}

After obtaining \( r \), if \( r \) is even and \( a^{r/2} \not\equiv -1 \pmod{N} \), compute \( \gcd(a^{r/2} \pm 1, N) \) to find nontrivial factors of \( N \). The quantum circuit for Shor's algorithm, comprising initialization, modular exponentiation, inverse QFT, and measurement, is as follows:

\[
\Qcircuit @C=1em @R=1em {
    \lstick{\ket{0}^{\otimes n}} & \gate{H^{\otimes n}} & \multigate{1}{U_a} & \gate{\text{QFT}^{-1}} & \meter \\
    \lstick{\ket{0}^{\otimes m}} & \qw & \ghost{U_a} & \qw & \qw
}
\]

\subsubsection*{Comparison of Shor and Regev's algorithm}

\begin{table}[h]
\centering
\caption{Comparison of various aspects of Shor’s and Regev’s algorithms}
\begin{tabular}{|>{\raggedright\arraybackslash}m{3.5cm}|>{\raggedright\arraybackslash}m{5cm}|>{\raggedright\arraybackslash}m{5cm}|}
\hline
\textbf{Stage} & \textbf{Regev’s Algorithm} & \textbf{Shor’s Algorithm} \\ \hline
Initialisation &
$\sum_{z \in \{-D/2,\ldots,D/2-1\}^d} \rho_R(z) |z\rangle$ (discrete Gaussian over $\mathbb{Z}^d$) &\makecell[l]{$\sum_{z=0}^{2^n-1} |z\rangle$ (uniform superposition over $\mathbb{Z}$)} \\ \hline
Modular Exponentiation &
$f(z) = \prod_{i=1}^d a_i^{z_i+D/2} \bmod N$ (using the square-and-multiply method, leveraging multiplication of small integers) &
$f(z) = a^z \bmod N$ (using the precomputation method) \\ \hline
Lattice &
$\mathcal{L}$ (the target lattice) &
$r\mathbb{Z}$ (integer multiples of the period of modular exponentiation) \\ \hline
Dual Lattice after Fourier Transform &
$\mathcal{L}^*$ &
$\frac{1}{r}\mathbb{Z}$ \\ \hline
Measurement Result &
Approximation of $v \in \mathcal{L}^*/\mathbb{Z}^d$ &
Approximation of $\frac{k}{r}$, where $k$ is an integer \\ \hline
Post-processing &
Generate the lattice $\mathcal{L}'$ from $m = d+4$ measured $\hat{v}_i$ and apply LLL to recover $\mathcal{L}$ with constant success probability &
Apply the continued fraction algorithm to recover $r$ with constant success probability \\ \hline
\end{tabular}
\label{tab:comparison1}
\end{table}

We compare the procedures of Shor’s and Regev’s algorithms, both serving as instances of the Hidden Subgroup Problem and sharing numerous similarities, differing slightly in implementation, as shown in Table~\ref{tab:comparison1}.

\section{Appendix B: Complexity of Regev's algorithm}\label{B}

In this section, we review the time and space complexity of Regev's algorithm. We then discuss how the choice of modules, which differs from Regev's original setting, affects the overall complexity.

\begin{enumerate}
    \item \textbf{Time Complexity.}
    
The time complexity consists of three parts. 
\begin{itemize}
    \item 
In the initial state preparation, due to simplifications, we only need to prepare the \( O(\log d) \) most significant qubits in each dimension. The remaining qubits can be prepared in a uniform superposition using Hadamard gates. Using the technique in~\cite{Grover02}, we require \( \mathrm{poly}(\log d) \) operations to compute the rotation in each dimension. The total circuit size is thus \(O\bigl(d(\log D + \poly(\log d))\bigr)\).

\item
The most computationally expensive part is the modular exponentiation. The core of Regev's algorithm lies in computing the product of a series of \( O(\log d) \)-bit numbers \( a_i \) pairwise in a binary tree manner. The complexity of this part can be computed recursively as \( T(m) = 2 T(m/2) + M(O(\log d) \cdot m/2) \), where \( T(m) \) represents the number of gates required for multiplying \( m \) \( O(\log d) \)-bit numbers, and \( M(k) \) denotes the number of gates needed to compute the product of two \( k \)-bit numbers.  

The complexity \( T(m) \) and the overall complexity depend on \( M(k) \), which is determined by the specific choice of the multiplication method. In Regev’s paper, fast multiplication~\cite{HvD21} is adopted, with \( M(k) = O(k \log k) \), yielding a complexity of \( O(d \log^3 d) \). Since \( \log D \) binary trees are computed, this requires \( O(\log D \cdot d \log^3 d) \) quantum gates. Accounting for \( \log D \) squaring and multiplication operations, along with their uncomputation, each requiring \( O(n \log n) \) operations, the total complexity is \( O(\log D \cdot (d \log^3 d + n \log n)) \).  \footnote{If we use other multiplication methods, i.e., with \( M(k) = O(k^2) \), the recursive result yields \( O(d^2 \log^2 d) \) and a total size of \( O(\log D \cdot (d^2 \log^2 d + n^2)) \). The details can be found in Appendix C. }For convenience, the analysis in the main text of this paper adopts fast multiplication.

\item
In the final step, we apply the quantum Fourier transform over \( \mathbb{Z}_D^d \) to the \( \ket{\mathbf{z}} \) register. The circuit size for this step is \( O(d \log D \cdot \log((\log D)/\varepsilon)) \) if using approximate QFT with error \( \varepsilon \)~\cite{Copper02}. Taking \( \varepsilon = 1/\poly(d) \), the circuit size becomes \( O(d \log D \cdot (\log \log D + \log d)) \). Even when considering the standard QFT over \( \mathbb{Z}_D \) for d dimensions, the circuit size is \( O(d \log^2 D) \). 
\end{itemize}
Considering all three parts, the time complexity is dominated by the modular exponentiation. Following Regev’s optimal suggestion, taking \( \log D = O(\sqrt{n}) \) and \( d \approx \sqrt{n} \), the total complexity is \( O(n^{3/2} \log n) \). Specifically, it is dominated by \( O(\log D) \) multiplications and squaring operations on \( n \)-bit numbers, each requiring a circuit size of \( O(n \log n) \).
\item \textbf{Space Complexity.}

A significant challenge of Regev’s algorithm arises from its space complexity. The space complexity of Regev’s modular exponentiation part is \( O(n^{3/2}) \), resulting in an overall space complexity of \( O(n^{3/2}) \). This is significantly higher than Shor’s algorithm, which requires \( O(n) \).

The increase in space complexity is due to the irreversibility of the squaring operations. Specifically, in the square-and-multiply stage, we cannot implement an in-place squaring operation \( \ket{x} \xrightarrow{} \ket{x^2 \pmod{N}} \); instead, we must use out-of-place operations \( \ket{x}\ket{0} \xrightarrow{} \ket{x}\ket{x^2 \pmod{N}} \) and store all intermediate results \( \ket{x} \) after the squaring operations. These intermediate states are subsequently removed through uncomputation. Since \( \log D \) intermediate states are stored during the process, each requiring \( n \) qubits, this results in a total of \( O(n^{3/2}) \) ancilla qubits.
\end{enumerate}

\subsection*{Time complexity in different situations}
In this section, we analyze the computational complexity of Regev's algorithm under assumptions and settings different from those of Regev, as the initial state preparation and modular exponentiation steps we adopt differ from those in Regev’s approach.

\subsubsection*{Initialization}

In the initial state preparation, using the technique in~\cite{Grover02}, we require \( \mathrm{poly}(n) \) operations to compute the rotation of \( n \) qubits to prepare the initial state, with some ancilla qubits.

Due to simplifications, we only need to prepare the \( O(\log d) \) most significant qubits in each dimension, requiring \( \mathrm{poly}(\log d) \) operations. The remaining qubits can be prepared in a uniform superposition using Hadamard gates, necessitating a total of \( O(\log D) \) quantum gates. The total circuit size is \( d (\log D + \mathrm{poly}(\log d)) \) using~\cite{Grover02}.  

There are also other methods to prepare such an initial state, for example, the approach in~\cite{Long01}, which requires \( O(2^n \cdot n^2) \) operations to prepare an arbitrary \( n \)-qubit initial state, with a simpler structure and zero ancilla qubits. For the \( O(\log d) \) most significant qubits, we need \( O[\exp(O(\log d)) \cdot O(\log d)^2] = \mathrm{poly}(d) \) operations, and in total \( d (\log D + \mathrm{poly}(d)) \) using~\cite{Long01}.  

The circuit size using~\cite{Long01} is larger than that of Regev’s approach but simpler to implement. Note that the complexity of \( d \,\mathrm{poly}(d) \) may dominate the total complexity, depending on the degree of the polynomial \( \mathrm{poly}(d) \), which is determined by the constant factor associated with the number of \( O(\log d) \) qubits requiring precise preparation. Therefore, for asymptotic analysis, the complexity reduction in~\cite{Grover02} is essential. However, in practical experiments, a detailed analysis and comparison are required. Especially for small instances, the method in~\cite{Grover02} is often used due to the absence of ancilla qubits.

\subsubsection*{Modular exponentiation}

When considering the most computationally expensive part, the modular exponentiation, we note two points:

\textbf{1.} In the modular exponentiation scheme, Regev's algorithm must adopt the square-and-multiply method rather than precomputation in order to exhibit an asymptotic advantage in time complexity, while sacrificing space complexity. This represents a trade-off in Regev’s algorithm. However, when Regev’s algorithm also adopts the precomputation method for modular exponentiation, it may be easier to implement, but its asymptotic time and space complexities will be similar to those of Shor’s algorithm.

In the modular exponentiation of Regev's algorithm, we need to compute
\[
\prod_{i} a_i^{z_i + D/2} \pmod{N},
\].

The square-and-multiply algorithm first computes \( \prod_{i} a_i^{z_{i0}} \bmod N \) for the most significant binary bit $0$. The result is then squared and multiplied by the term corresponding to the next binary bit, \( \prod_{i} a_i^{z_{i1}} \bmod{N} \), squared, and so on, until all bits indexed by \( j \) are multiplied. Here, \( z_{ij} \) denotes the \( j \)-th bit of \( z_i + D/2 \), where 
\( j \in \{0, 1, \ldots, \log D - 1\} \), and \( j = 0 \) denotes the most significant bit.\footnote{For notational convenience in the subsequent discussion, we sometimes set the most significant bit as \( j = \log D - 1 \).}

As described in the main text, such an algorithm can lead to a reduction in time complexity. This algorithm requires the implementation of modular multiplication between \( a_i \), as well as modular squaring operations on the computational register. However, in quantum computing, the squaring operator cannot be implemented in-place, thus requiring a significant number of ancilla qubits to store intermediate states (as mentioned above, and the solution proposed in the theoretical section is aimed at mitigating this issue). Furthermore, after all operations are completed, we need to uncompute the intermediate states, which typically doubles the circuit size.

To avoid squaring operations, we can adopt a precomputation scheme similar to the modular exponentiation in Shor's algorithm (sometimes referred to as "fast exponentiation" in other works). In this method, each block first performs the multiplication of \( a_i^{2^j \cdot z_{ij}} \) for all \( i \), and then multiplies the result into the computational register.

Implementing such a circuit avoids the need for squaring operations, replacing them with multiplications between \( a_i^{2^j} \) and \( a_2^{2^j} \) for each \( j \), followed by multiplying the result into the computational register. However, since \( a_i^{2^j} \) grows rapidly with \( j \) and is generally no longer a small prime number, performing a binary tree-style multiplication between them does not reduce the circuit size. In this case, the asymptotic complexity of this approach is identical to that of the direct multiplication method, where \( a_i^{2^j} \) is sequentially multiplied in the computational register. This can be precomputed and compiled into circuits before the quantum circuit's construction. This approach is similar to that in Shor's algorithm.

In practice, this requires constructing a total of \( d \cdot \log D \approx O(n) \) distinct multiplications, which is asymptotically equivalent to the number of multiplications required by Shor's algorithm. Shor's algorithm can be regarded as a special case where \( d = 1 \) and \( \log D \approx O(n) \). In this situation, we do not know if Regev's approach has other advantages over Shor's implementation. (Considering the various variants of Shor's algorithm, its practical implementation may be more efficient.) One possible advantage to consider is that, after dimensionality increase, the precision requirements for the quantum Fourier transform might decrease.

Although it no longer offers an advantage in asymptotic complexity, adopting this scheme significantly simplifies the experimental implementation of Regev's algorithm, making it suitable for proof-of-principle experiments, in which we only need to precompute \( a_i^{2^j} \) and compile them into the classical--quantum multiplier.

Figures~\ref{fig:architecture1} to~\ref{fig:architecture3} illustrate the modular exponentiation architectures under these two distinct schemes for a small case d=2. (potential ancilla qubits and the subsequent quantum Fourier transform are ignored). and In the Appendix E, we will present the complete experimental circuits for both schemes based on Figures~\ref{fig:architecture1} and~\ref{fig:architecture3}.

\textbf{2.} In the modular exponentiation stage, the choice of the specific multiplier affects the final complexity result, but does not change the fact that the most time-consuming parts of the algorithm remain the \( n \)-bit multiplication and squaring. It also does not alter the optimization factor \( n^{1/2} \) of Regev's algorithm compared to Shor's algorithm.

The complexity of computing the product of a series of \( O(\log d) \)-bit numbers \( a_i \) pairwise in a binary tree manner can be computed recursively as \( T(m) = 2 T(m/2) + M(O(\log d)\cdot m/2) \), where \( T(m) \) represents the number of gates required for multiplying \( m \) \( O(\log d) \)-bit numbers, and \( M(k) \) denotes the number of gates needed to compute the product of two \( k \)-bit numbers.

The complexity \( T(m) \) and the overall complexity depend on \( M(k) \), which is determined by the specific choice of the multiplication module. In Regev’s paper, fast multiplication from~\cite{HvD21} is adopted, with \( M(k) = O(k \log k) \), yielding a complexity of \( O(d \log^3 d) \). Since \( \log D \) binary trees are computed, this requires \( O(\log D \cdot d \log^3 d) \) quantum gates. Accounting for \( \log D \) squaring and multiplication operations, along with their uncomputation, each requiring \( O(n \log n) \) operations, the total complexity is \( O(\log D \cdot (d \log^3 d + n \log n)) \). If we use other multiplication such as the schoolbook, which is more efficient for practical problems, with \( M(k) = O(k^2) \), we can recursively compute that each binary tree multiplication requires \( O(d^2 \log^2 d) \) quantum gates. Considering that each squaring, multiplication, and their uncomputation operations require \( O(n^2) \) gates, the total complexity becomes \( O(\log D \cdot (d^2 \log^2 d + n^2)) \). Following Regev’s optimal suggestion, taking \( \log D = O(\sqrt{n}) \) and \( d \approx \sqrt{n} \), the total complexity is \( O(n^{3/2} \cdot \log n) \) using fast multiplication or \( O(n^{5/2}) \) using schoolbook multiplication.
    
Considering all three parts, if the complexity of the initial state preparation does not exceed \( O(n^{3/2} \log n) \) or \( O(n^{5/2}) \), the time complexity is dominated by the \( n \)-bit multiplication and squaring operations.

A comparison of the time complexity for different modular exponentiation schemes and multiplier choices is detailed in Table~\ref{tab:comparison2}, which also includes variants of Regev’s algorithm.

\begin{table}[ht]
\centering
\setlength{\tabcolsep}{2pt} 
\begin{tabular}{|m{2cm}|m{5cm}|m{3cm}|m{2cm}|m{3cm}|}
\hline
Algorithm & Modular Exponentiation Scheme & Multiplier Choice & Circuit Size & Number of Qubits \\
\hline
\multirow{2}{*}{Shor}
  & \multirow{2}{*}{Precompute \(a^{2^j}\)}
    & Fast Multiplication & \( O(n^2 \log n) \) & \( O(n \log n) \) \\
  \cline{3-5}
    & & Schoolbook & \( O(n^3) \) & \( O(n) \) \\
\hline
\multirow{4}{*}{Regev}
  & \multirow{2}{*}{Square-and-Multiply}
    & Fast Multiplication & \( O(n^{3/2} \log n) \) & \( O(n^{3/2}) \) \\
  \cline{3-5}
    & & Schoolbook & \( O(n^{5/2}) \) & \( O(n^{3/2}) \) \\
  \cline{2-5}
  & \multirow{2}{*}{Precompute \(a_i^{2^j}\)}
    & Fast Multiplication & \( O(n^2 \log n) \) & \( O(n \log n) \) \\
  \cline{3-5}
    & & Schoolbook & \( O(n^3) \) & \( O(n) \) \\
\hline
\end{tabular}
\caption{Comparison of quantum algorithms for different modular exponentiation schemes and basic multiplier choices. Unless otherwise specified, Fast Multiplication refers to~\cite{HvD21}, and Schoolbook refers~\cite{Ric18}.}
\label{tab:comparison2}
\end{table}

\section{Appendix C: Time and space complexity under Intermediate Uncomputation method}\label{C}

In this appendix, we provide some proofs of the theorems stated in the main text.

For clarity of exposition, we use the term \emph{squarings} to refer to the number of mathematical squarings required in the underlying task (for Regev's algorithm, this number is \( m = \log D - 1 \)), and we use \emph{squaring operations} to refer to the number of actual squaring gates in the quantum circuit, including inverse squarings used for uncomputation. Throughout this appendix, the time complexity is measured in squaring operations, and the space complexity is measured in the number of \( n \)-bit registers.

\begin{theorem}[Simple Strategy]
Using the notation and simple strategy defined in the main text, the number of registers required to compute \( m \) squarings and the corresponding number of squaring operations are given by
\begin{align}
S_{\mathrm{simple}}(m,k)
  &= \left\lceil \frac{m+1}{k} \right\rceil + k - 1, \\
T_{\mathrm{simple}}(m,k)
  &= \Bigl( 2\Bigl\lceil \frac{m+1}{k} \Bigr\rceil - 1 \Bigr)(2k-1)
     + 2\Bigl( (m+1) - k \Bigl\lceil \frac{m+1}{k} \Bigr\rceil\Bigr)-4 .
\end{align}
In particular, \(T_{\mathrm{simple}}(m,k) = O(m)\); more precisely, one has
\[
  T_{\mathrm{simple}}(m,k) \le 4m-2 .
\]
\end{theorem}

\begin{proof}
The closed-form expressions for \( S_{\mathrm{simple}}(m,k) \) and \( T_{\mathrm{simple}}(m,k) \) follow from a straightforward step-by-step analysis of the simple strategy and are therefore omitted here. For convenience, one may first count the total number of multiplications appearing in the schedule and then subtract \(4\), since in the simple strategy the first register is always involved in four multiplications (and their inverses) that contribute only a constant overhead.

For the upper bound on the time cost, the intuition is as follows. In the direct method, each ancilla register (except for the final output register) participates in exactly two squaring operations, while in the simple strategy each ancilla register participates in at most four squaring operations (the final output register is still used only once). Hence the total number of squaring operations increases by at most a factor of 2 compared to the direct method, which uses \( 2m - 1 \) squarings, and we obtain \( T_{\mathrm{simple}}(m,k) \le 4m - 2 \). The same bound can also be derived rigorously from the analytic expression of \( T_{\mathrm{simple}}(m,k) \) by elementary bounding and simplification, and we omit the details.
\end{proof}

As an immediate consequence of this theorem, we obtain the following corollary.

\begin{corollary}[Simple Strategy (Space-Reduced, Linear-Time)]
The number of registers required to compute \( m \) squarings can be reduced to \( O(\sqrt{m}) \), while the time complexity remains \( O(m) \), increasing by at most a constant factor not exceeding 2 compared to direct computation.
\end{corollary}

We now turn to the proof of the space lower bound. For this purpose, we first introduce the following lemma.

\begin{lemma}\label{lem:2nminus1}
Using a total of \( n \) registers, one cannot compute the result of the \((2^{n}-1)\)-th squaring, even if one does not insist on uncomputing all ancilla registers at the end.
\end{lemma}

Using this lemma, we can now prove the space lower bound stated in the main text.

\begin{theorem}[Space Lower Bound]
Using \( n \) registers, one cannot complete the computation of \( 2^{n-1} \) squarings (together with their uncomputation). Consequently, for \( m \) squarings, at least \( \lceil \log_2(m+1) \rceil + 1 \) registers are required. This establishes a lower bound on the space complexity.
\end{theorem}

\begin{proof}
Fix integers \( n \ge 1 \) and \( G \ge 1 \), and consider the following two statements:
\begin{description}
\item[(A\(_n(G)\))] Using \( n \) registers, one can compute the result of the \(G\)-th squaring (without requiring final uncomputation of all ancillas).
\item[(B\(_n(G)\))] Using \( n+1 \) registers, one can complete the computation of \(G+1\) squarings together with full uncomputation (all ancillas are restored to their initial states at the end).
\end{description}

We claim that, for every \( n,G \), \((\mathrm{A}_n(G))\) and \((\mathrm{B}_n(G))\) are equivalent.

\medskip
\noindent\emph{(\(\mathrm{A}_n(G) \Rightarrow \mathrm{B}_n(G)\)).}
Assume \((\mathrm{A}_n(G))\). Using \( n+1 \) registers, we first run this computation on the first \( n \) registers to obtain the \(G\)-th squaring. We then perform one additional squaring and store the result in the \((n+1)\)-st register, obtaining the \((G+1)\)-st squaring. Finally, we reverse the first part of the computation on the first \( n \) registers, leaving the \((n+1)\)-st register untouched. This restores the first \( n \) registers to their initial states and yields \(G+1\) squarings with full uncomputation, proving \((\mathrm{B}_n(G))\).

\medskip
\noindent\emph{(\(\mathrm{B}_n(G) \Rightarrow \mathrm{A}_n(G)\)).}
Conversely, assume \((\mathrm{B}_n(G))\). Consider any computation that, using \( n+1 \) registers, performs \(G+1\) squarings and then uncomputes all ancillas. Let \( t \) be the first time step at which the \((G+1)\)-st squaring is present in some register. At the previous step \( t-1 \), the result of the \(G\)-th squaring is already stored, and at least one register is still idle (since the \((G+1)\)-st squaring has not yet been written). By discarding this idle register and all later operations (the final squaring and the subsequent uncomputation), we obtain a computation using at most \( n \) registers that produces the \(G\)-th squaring. Hence \((\mathrm{A}_n(G))\) holds.

\medskip
Thus, for all \( n \ge 1 \) and \( G \ge 1 \), \((\mathrm{A}_n(G))\) and \((\mathrm{B}_n(G))\) are equivalent.

Now apply Lemma~\ref{lem:2nminus1} with \( G = 2^{n} - 1 \). The lemma states that \((\mathrm{A}_n(2^{n}-1))\) is false; hence, by the above equivalence, \((\mathrm{B}_n(2^{n}-1))\) is also false. Therefore, using \( n+1 \) registers, one cannot complete the computation of
\[
(2^{n}-1) + 1 = 2^{n}
\]
squarings together with full uncomputation.

Let \( r = n+1 \). Then, for every integer \( r \ge 2 \), using \( r \) registers one cannot complete the computation of \( 2^{r-1} \) squarings with full uncomputation. (The case \( r = 1 \) is trivial, since with a single register one cannot both perform a squaring and return to the initial state.)

Finally, suppose that \( m \) squarings can be completed using \( r \) registers (with full uncomputation). From the above impossibility, we must have
\[
  m \le 2^{r-1} - 1,
\]
since otherwise one could complete \( 2^{r-1} \) squarings. Equivalently,
\[
  2^{r-1} \ge m + 1.
\]
Taking base-2 logarithms and rearranging gives
\[
  r - 1 \ge \log_2(m+1) \quad\Rightarrow\quad
  r \ge \log_2(m+1) + 1.
\]
Thus any algorithm that completes \( m \) squarings (with full uncomputation) must use at least
\[
  r_{\min} = \bigl\lceil \log_2(m+1) \bigr\rceil + 1
\]
registers, which proves the claimed lower bound on the space complexity.
\end{proof}

We now prove that we can achieve the aforementioned space lower bound through a recursive strategy.

\begin{definition}[Maximum squarings with \( n \) registers]
Let \( G(n) \) denote the maximum number of squarings that can be completed using \( n \) registers, including all necessary uncomputations. Specifically, \( G(n) \) is the number of squarings such that the intermediate results of all squarings can be uncomputed at the end. In our squaring computation chain, the first computation result does not involve any squaring and can always be computed. This result is defined as the \( 0 \)-th squaring result. We immediately observe that \( G(1) = 0 \), \( G(2) = 1 \), and \( G(3) = 3 \), as demonstrated in the example in the main text.
\end{definition}

We now prove that \( G(n) = 2^{n-1} - 1 \) by induction.

\begin{theorem}[Binary Recursion (Achieving the Space Lower Bound)]
The binary recursion strategy achieves the result \( G(n) = 2^{n-1} - 1 \).
\end{theorem}

\begin{proof}
We prove this by induction.

\textbf{Base case:} For \( n = 2 \), we have \( G(1) = 1 \), which is consistent with the definition.

\textbf{Inductive step:} Assume that the statement holds for all \( k \le n-1 \), i.e., \( G(k) = 2^{k-1} - 1 \). We will prove that \( G(n) = 2^{n-1} - 1 \).

We construct the recursive strategy as follows: First, we use \( n-1 \) registers to complete \( G(n-1) \) squarings and uncompute all intermediate results. At this point, only the result of the \( G(n-1) \)-th squaring is held in the registers. We then retain this result as the starting point for the next phase of computation.

In the next phase, the next squaring can either be computed or uncomputed at any point during the process. Hence, it effectively serves as the initial result of the previous phase (which was the \( 0 \)-th result). We can then use \( n-1 \) registers to perform \( G(n-1) + 1 \) additional squarings while simultaneously uncomputing all intermediate results. After this, we are left with only the results of the \( G(n-1) \)-th and \( 2G(n-1) + 1 \)-th squarings being held in the registers. Finally, we uncompute the \( G(n-1) \)-th result, thus completing all computations.

Therefore, we have the recurrence relation:
\[
G(n) \ge 2G(n-1) + 1.
\]

By the inductive hypothesis, we know that \( G(n-1) = 2^{n-2} - 1 \), so:
\[
G(n) \ge 2(2^{n-2} - 1) + 1 = 2^{n-1} - 1.
\]

Thus, we have shown that \( G(n) \ge 2^{n-1} - 1 \).

Finally, combining this with the space lower bound, we conclude that \( G(n) = 2^{n-1} - 1 \), completing the proof.
\end{proof}

We now turn to the time complexity. Let \( T_G(n) \) denote the total number of multiplications and inverses performed during the recursive process 
For convenience, we count also the multiplications acting on the ``0-th'' register (the initial input), so \( T_G(n) \) is slightly larger than the number of actual squarings and large integer multiplications. As concrete examples, one easily checks that
\[
T_G(2) = 3, \qquad T_G(3) = 9 ,
\]
consistent with the constructions given in the main text. Specifically, \( T_G(n) \) counts the total number of multiplication operations required. We will now prove that \( T_G(n) \le 3^{n-1} \).

\begin{definition}[Time cost of binary recursion]
Let \( T_G(n) \) denote the total number of multiplications (and their inverses) performed by the binary recursion strategy when using \( n \) registers to realize \( G(n) \) squarings together with full uncomputation of all ancillas. 
\end{definition}

\begin{theorem}[Binary Recursion (Time Cost)]
The time cost of the binary recursion strategy is \( T_G(n) \le 3^{n-1} \).
\end{theorem}

\begin{proof}
We prove this by induction.

\textbf{Base case:} For \( n = 2 \), we have \( T_G(2) = 3 \), which is consistent with the definition.

\textbf{Inductive step:} Assume the statement holds for all \( k \le n-1 \), i.e., \( T_G(k) \le 3^{k-1} \). We now prove that \( T_G(n) \le 3^{n-1} \).

During the \( n \)-th recursive step, we first perform a squaring operation for \( G(n-1) \) (which requires \( T_G(n-1) \) operations). Then, we perform another squaring operation for \( G(n-1) \), and finally, to uncompute the result of \( G(n-1) \), we effectively perform another \( G(n-1) \)-like operation. Therefore, the total time required is at most three times the previous time complexity:
\[
T_G(n) \le 3 \cdot T_G(n-1).
\]
Thus, by the inductive hypothesis, we have:
\[
T_G(n) \le 3 \cdot 3^{n-2} = 3^{n-1}.
\]

This completes the proof.
\end{proof}

The quantity \( T_G(n) \) bounds from above the actual number of large-integer squarings (or multiplications) needed by the binary recursion strategy, because our counting also includes the multiplications acting on the initial register. Let \( T(m) \) denote the number of large-integer multiplications or squarings needed to realize \( m \) squarings under the binary recursion strategy. Whenever \( m = G(n) = 2^{n-1} - 1 \), we have
\[
  T(m) \le T_G(n) \le 3^{\,n-1}.
\].

For small values of \( n \), explicit constructions give
\[
  T_G(3) = 9,\quad T_G(4) = 25,
\]
and, correspondingly, the actual numbers of squarings satisfy
\[
  T(3) = 5,\quad T(4) = 17,
\]
all strictly below the general upper bound \( 3^{\,n-1} \). In practice, tighter bounds (or optimal schedules) can be sought via search algorithms, but for our asymptotic analysis the above upper bound suffices. 

Using the relation \( m+1 = 2^{n-1} \), we can rewrite the bound as
\[
  T(m) \le 3^{\log_2 (m+1)} = (m+1)^{\log_2 3},
\]
which is consistent with the \( O(m^{\log_2 3}) \) time complexity stated in the main text.

\begin{corollary}[Binary recursion (achieving the space lower bound)]
The space-complexity lower bound described above for \( m \) squarings can be achieved using the binary recursion strategy, with the time complexity not exceeding \( O(m^{\log_2 3}) \).
\end{corollary}

We finally turn to the \( k \)-ary recursion, which is a generalization of binary recursion. In the proof of binary recursion, we clean and uncompute intermediate states after every two stages, leaving only one register at the end; this process is called binary recursion. If, instead of two stages, we clean and uncompute all intermediate states after every \( k \) stages, we obtain the generalization to \( k \)-ary recursion.

In \( k \)-ary recursion, the number of squarings that can be computed grows at a rate approximately proportional to \( k^n \) with respect to the number of registers \( n \). Meanwhile, the number of squaring (and inverse) operations grows at a rate approximately proportional to \( (2k-1)^n \). These observations lead to the following conclusion:

\begin{theorem}[k-ary recursion]
Using a \( k \)-ary recursion strategy, the space complexity is \( O(\log m) \), and the time complexity is \( O\!\bigl(m^{\log_k(2k-1)}\bigr) \).
\end{theorem}

From this, we derive the following corollary:

\begin{corollary}[k-ary recursion (log-space, near-linear time)]
The asymptotic time complexity can be reduced to \( O(m^{1 + \epsilon}) \) using a \( k \)-ary recursion strategy, where \( \epsilon \) is a positive constant given by
\[
  \epsilon = \log_k (2k-1) - 1.
\]
At the same time, the space complexity is kept at \( O(\log m) \). Moreover, \(\epsilon\) can be made arbitrarily small by choosing \(k\) sufficiently large.
\end{corollary}

In fact, we can provide explicit expressions for the space and time complexities of the recursive strategy. To do this, we first fix the basic recursive starting unit, namely the chosen base unit \( x_0 \), which uses \( n_0 \) additional registers to compute \( x_0 \) squarings in time \( t_0 \). Typically, we choose \( x_0 = 1 \), \( n_0 = 1 \), and \( t_0 = 1 \).

Fixing \( x_0 \), \( n_0 \), and \( t_0 \) according to this base unit, and assuming that \( \log_k\!\left(\frac{m+1}{x_0}\right) \) is a positive integer, we obtain the following expressions for the space and time costs:
\begin{align}
S_{\mathrm{rec}}(m;k,x_0) &= n_0 + (k-1)\,\log_k\!\left(\frac{m+1}{x_0}\right), \\
T_{\mathrm{rec}}(m;k,x_0) &\le t_0 \left(\frac{m+1}{x_0}\right)^{\log_k(2k-1)}.
\end{align}
These expressions can be used as asymptotic estimates for the actual space and time costs.

In the above formulas, we can consider a special variant, where instead of fixing \( k \), we choose the recursion depth \( \ell \) to be a fixed constant and set
\[
  k \approx \Bigl(\frac{m+1}{x_0}\Bigr)^{1/\ell}.
\]
As the problem size grows, this approach leads to linear time complexity and sublinear space complexity. In this case, we have approximately
\[
  S \approx (k-1)\,\ell, \qquad T \approx (2k-1)^{\ell}.
\]
Since \( \ell \) is a fixed constant and \( k \approx \bigl((m+1)/x_0\bigr)^{1/\ell} \), we obtain
\[
  S = O\!\bigl(\sqrt[\ell]{m}\bigr), \qquad T = O(m).
\]

This leads to the following conclusion.

\begin{corollary}[Variable-Arity Recursion (Linear Time, Sublinear Space)]
For any integer \( \ell \ge 1 \), there exists a variable-arity recursion strategy that computes \( m \) squarings using only \( O\!\left(\sqrt[\ell]{m}\right) \) registers, while keeping the time complexity at \( O(m) \). This strategy can be realized by choosing the recursion arity as \( k \approx \sqrt[\ell]{m + 1} \).
\end{corollary}

For practical problems, the above variable-arity recursion strategy also fixes an effective \( k \) for a given problem size \( m \), so it still belongs to the same family of strategies as the \( k \)-ary recursion, just viewed from a different perspective.

Additionally, note that the simple strategy is actually a special case of \( \ell = 2 \). In this strategy, we often achieve a relatively balanced practical performance.

Finally, by substituting \( m = \log D - 1 \approx C \sqrt{n} \), we can derive all the computational complexity results for the Regev algorithm presented in the main text.

\section{Appendix D: Choice of parameters}\label{D}
In Regev's theoretical analysis, the parameters are required to satisfy the following conditions, which ensure that the algorithm succeeds with at least constant probability:

\begin{itemize}
    \item \(\boldsymbol{d}\): The lattice dimension parameter. In the main setting, \(d\) is chosen as \(\approx\sqrt{n}\), which balances the complexity of the quantum part and the success probability of the lattice reduction step. 
    
    \item \(\boldsymbol{R}\): The Gaussian width parameter for preparing the discrete Gaussian superposition. The quantum procedure requires \( R > \sqrt{2d} \). To ensure that classical lattice reduction (e.g., LLL) recovers the target short vector with high probability, one needs
    \[
        R > 2^{d + n/d} \, T,
    \]
    where \( T = \exp(c n/d) \) is a heuristic upper bound on the norm of a nontrivial vector in \( L \setminus L_0 \), with \( c > 0 \) a constant. Substituting \( d = \sqrt{n} \) gives \( R \approx \exp(C \sqrt{n}) \).
    
    \item \(\boldsymbol{D}\): The discretization parameter for the output of the quantum Fourier transform. It is chosen as a power of two satisfying
    \[
        2\sqrt{d} \, R \leq D < 4\sqrt{d} \, R.
    \]
    Under the main setting \( d = \sqrt{n} \) and \( R \approx \exp(C\sqrt{n}) \), we have
    \[
        D \approx \Theta({n}^{1/4} \, \exp(C\sqrt{n})), \quad \log D = \Theta(\sqrt{n}).
    \]
    This is used for further analysis of computational complexity.

    \item \(\boldsymbol{\delta}\): The distance bound between the measured vector \(w\) and the true coset representative \(v \in L^*\) after the quantum measurement. With high probability,
    \[
        \delta \approx \frac{\sqrt{d}}{\sqrt{2}R},
    \]
    hence larger \(R\) leads to smaller measurement error.

    \item \(\boldsymbol{S}\): The embedding scaling parameter in the enlarged lattice \(\mathcal{L'}\) used for post-processing. It is typically chosen as \(S = \delta^{-1}\), balancing the requirement that short vectors in \(\mathcal{L'}\)  correspond to valid vectors in \(\mathcal{L}\) and that the embedding does not excessively amplify the error coordinates.

    \item \(\boldsymbol{m}\): The number of samples. if $m \geq d+4$, then with probability at least $1/4$ the sampled vectors $w_1,\dots,w_m$ generate the entire lattice \(\mathcal{L^*}/\mathbb{Z}^d\). In practice, we typically fix $m=d+4$ in the algorithm.

\end{itemize}
According to the parameter selection rules described above, the parameters in experiment are chosen as follows:
The integer to be factored is \( N=35 \), which has \( n=6 \) bits. Therefore, the computational register requires \( n=6 \) qubits, corresponding to the six qubits \( q_6 - q_{11} \) in Figures~\ref{fig:architecture1} to~\ref{fig:architecture3}.
Other parameters are selected below:

\begin{enumerate}

    \item \(\boldsymbol{d}\): 
    \(d\) is chosen as \(\sqrt{n} \approx 2\). Correspondingly, the \(b_i\) are selected as the first two small primes, \(b_1 = 2, b_2 = 3\), and \(a_i = b_i^2\).
    
    \item \(\boldsymbol{R}\): For small-scale use, \( R > \sqrt{2d} \) is sufficient. The requirement \( R > 2^{d + n/d} \, T \), where \( T \) is a parameter defined in a heuristic assumption, serves as an asymptotic lower bound, which can be ignored in small-scale examples. In this example, we require \( R > \sqrt{2 \cdot 2} = 2 \), and we can choose \( R \) as a number slightly larger than 2.

    \item \(\boldsymbol{D}\): 
    It is chosen as a power of two satisfying
    \[
        2 \sqrt{d} \, R \; \le \; D \; < \; 4 \sqrt{d} \, R.
    \]
    Here, \(\log D\) represents the number of qubits required for each dimension of the lattice. In this example, \(D=8\) is sufficient. Therefore, each dimension of the lattice requires 3 qubits, corresponding to qubits \(q_0 - q_2\) and \(q_3 - q_5\) in Figures~\ref{fig:architecture1} to~\ref{fig:architecture3}.

    \item \(\boldsymbol{\delta}\): 
    \(\delta\) characterizes an upper bound on the error in lattice operations and is an implicit parameter, roughly given by
    \[
        \delta \approx \frac{\sqrt{d}}{\sqrt{2} R}.
    \]
    Hence, \(\delta\) is chosen as \(\frac{1}{R} \approx \frac{1}{2}\), which also determines the size of \(S\) in the post-processing step.

    \item \(\boldsymbol{S}\): 
    \(S\) is used for scaling in the post-processing step and is typically chosen as \(S = \delta^{-1}\), so we set \(S = 2\). In practice, this choice represents a trade-off: selecting a larger or smaller \(S\) may also allow finding suitable elements. In our experiments, for convenience in the classical post-processing, we directly set \( S = 8 \) so that \( \mathcal{L} \) becomes an integer lattice. Equivalently, one could take \( S = 2 \) and rescale the entire lattice by a factor of four, which would also lead to correct results.

    \item \(\boldsymbol{m}\): 
    The parameter \( m \) is typically fixed as \( m = d + 4 \) to guarantee a success probability bounded below by a constant, although smaller values of \( m \) may still succeed with non-negligible probability. In our experiments, we set \( m = 4 \).

\end{enumerate}

For small-scale experiments, the parameter selection exhibits flexibility. For
example, requiring \( D \) to be slightly larger than \( R \) ensures that the Gaussian distribution is approximately truncated at \( D \), facilitating error analysis. In practical applications, selecting a larger \( R \) such that it exceeds \( D \) increases truncation error but reduces the peak width \( \delta \) after the Fourier transform, potentially still enabling successful factorization. In fact, for very small-scale factorization, setting \( R \to \infty \) is sufficient, equivalent to preparing the initial state in a uniform superposition, as discussed in the experiment section. However, for slightly larger numbers, theoretical proof or simulation verification is still required.

\section{Appendix E: Experimental ciruit structure}\label{E}

\subsection*{Different architectures of modular exponentiation in Regev’s algorithm}

If controlled multiplication and squaring are regarded as the building blocks, the quantum circuit is illustrated in Figure~\ref{fig:architecture1}, where the potentially required ancilla qubits are ignored.

The circuit in Figure~\ref{fig:architecture1} requires the implementation of modular multiplication of \( a_1 \) and \( a_2 \), as well as modular squaring operations on the computational register. However, in quantum computing, the construction of a squaring operator introduces challenges. Since modular squaring operations are irreversible, they cannot be implemented in-place, thus requiring a significant number of ancilla qubits to store intermediate states. Furthermore, after all operations are completed, we need to uncompute the intermediate states, which typically doubles the circuit size.

\begin{figure}[ht]
\centering
\includegraphics[width=0.7\textwidth]{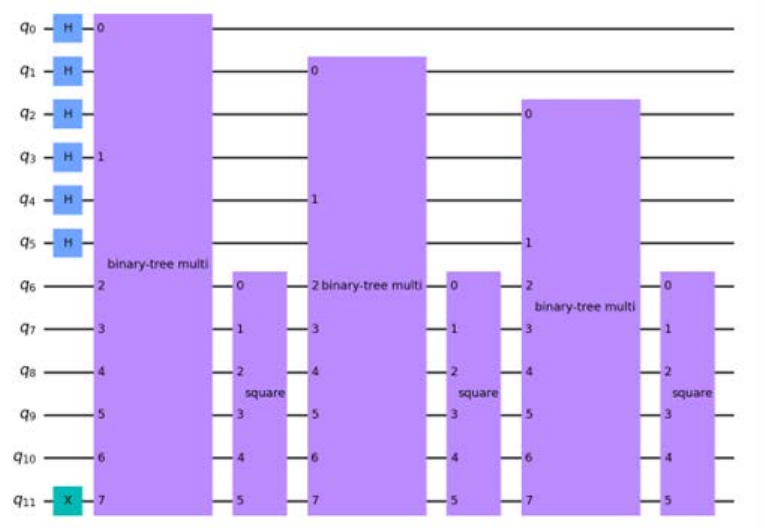}
\caption{Square-and-Multiply Method: The control registers \( q_0, \ldots, q_2 \) and \( q_3, \ldots, q_5 \) represent the two-dimensional variables \( (z_1, z_2) \), while the computational register \( q_6, \ldots, q_{11} \) holds the result of the modular exponentiation, initialized to \(\ket{1}\) by an X gate.}
\label{fig:architecture1}
\end{figure}

We can adopt a scheme similar to the modular exponentiation in Shor's algorithm to avoid squaring operations (referred to as "fast exponentiation" in other works). The corresponding circuit is illustrated in Figure~\ref{fig:architecture2}.

\begin{figure}[ht]
\centering
\includegraphics[width=0.5\textwidth]{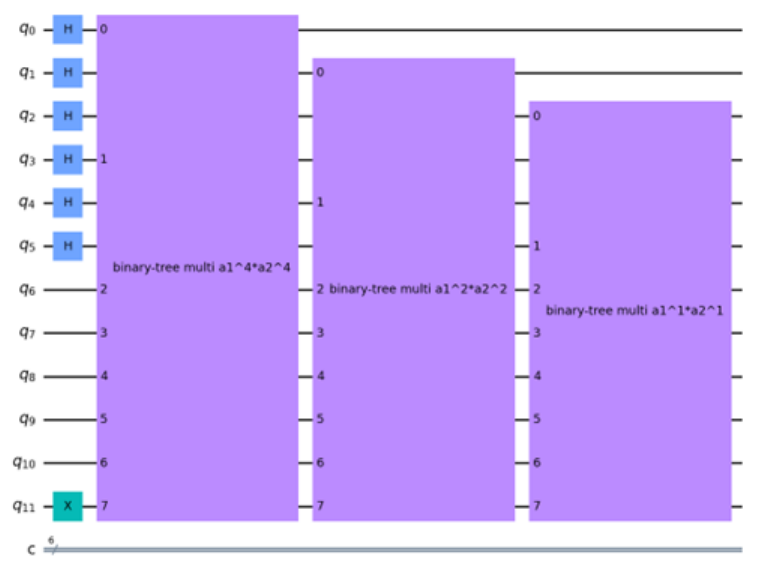}
\caption{Precomputation Method 1: Direct Computation of \( \prod_{i} a_i^{2^j}\bmod N \), Avoiding Squaring Operations}
\label{fig:architecture2}
\end{figure}

In this method, each block first performs the multiplication between \( a_1^{2^j\cdot z_{1j}} \) and \( a_2^{2^j \cdot z_{2j}} \), and then multiplies the result into the computational register. 

Implementing such a circuit avoids the need for squaring operations, requiring only a single multiplication between \( a_1^{2^j} \) and \( a_2^{2^j} \) for each \( j \), followed by multiplying the result into the computational register. However, since \( a_1^{2^j} \) and \( a_2^{2^j} \) are generally no longer small prime numbers, performing a binary tree-style multiplication between them does not reduce the circuit size, making it unnecessary. The asymptotic complexity of this approach is identical to that of the direct multiplication method shown in Figure~\ref{fig:architecture3}.

\begin{figure}[h!]
\centering
\includegraphics[width=0.65\textwidth]{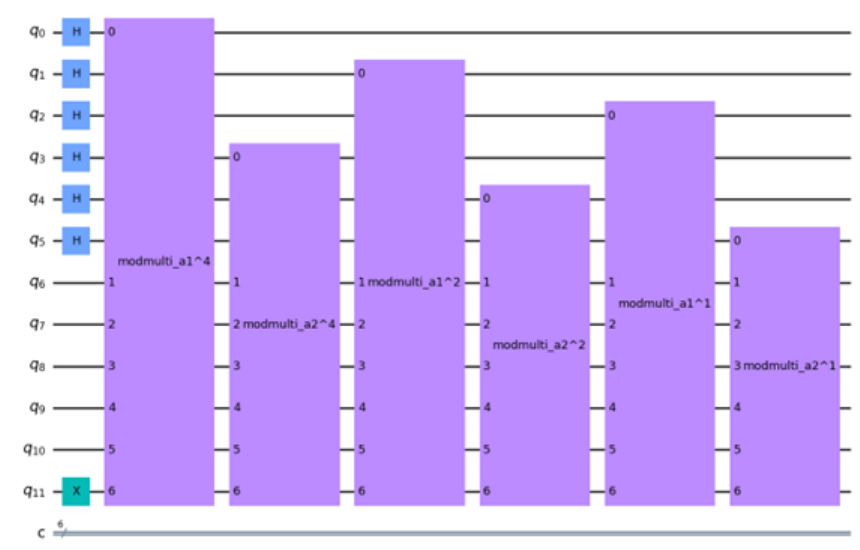}
\caption{Precomputation Method 2: Direct Multiplication Approach}
\label{fig:architecture3}
\end{figure}

The scheme in Figure~\ref{fig:architecture3} finally decomposes the modular exponentiation into six distinct controlled modular multiplications, where \( a_1^{2^j} \) and \( a_2^{2^j} \) can be precomputed and compiled into circuits before the quantum circuit's construction. This approach is similar to that in Shor's algorithm. Adopting this scheme significantly simplifies the experimental implementation of Regev's algorithm but sacrifices the advantage in asymptotic complexity, making it suitable only for proof-of-principle experiments. In practice, this requires constructing a total of \( d \cdot \log D \approx O(n) \) distinct multiplications, which is asymptotically equivalent to the number of multiplications required by Shor's algorithm. Shor's algorithm can be regarded as a special case where \( d = 1 \) and \( \log D \approx O(n) \). (Considering the various variants of Shor's algorithm, its practical implementation may be more efficient.) 

Figures~\ref{fig:architecture1} to~\ref{fig:architecture3} illustrate the modular exponentiation architectures under these two distinct schemes (potential ancilla qubits and the subsequent quantum Fourier transform are ignored). In the Experiment sections, we present the complete experimental circuits for both schemes based on Figures~\ref{fig:architecture1} and~\ref{fig:architecture3}.

\section{Appendix F: Compiled multiplication techniques}\label{F}

As described earlier, both modular exponentiation schemes ultimately reduce the modular exponentiation operation to computing several modular multiplications (and modular squarings).

A modular multiplication based on general-purpose design refers to a multiplication that can be implemented universally for a given modulus \( N \) and bit length \( n \), without requiring additional prior knowledge. It typically consists of basic adders and comparators, requiring a circuit size of \( O(n^3) \) and \( O(n) \) ancilla qubits~\cite{Haner17}.

In many experimental studies, a compiled multiplier is used, which requires fewer gate operations and no ancilla qubits but needs more sophisticated techniques. For smaller examples, we can easily verify the correctness of the multiplier, such as for \( 2^j \pmod{15} \), which is commonly used in Shor's algorithm demonstrations, as shown in Figure \ref{fig:multiplication_15}.\cite{Monz16}

\begin{figure}[ht]
\centering
\includegraphics[width=0.8\textwidth]{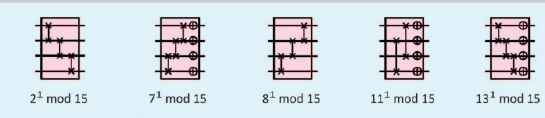}
\caption{Circuit implementation of the compiled multiplier for \( 2^j \pmod{15} \).}
\label{fig:multiplication_15}
\end{figure}

Here, we also apply the compiled multiplier to implement the multiplication for \( 4^j \pmod{35} \) and \( 9^j \pmod{35} \). To simplify the multiplier, we note that there are only 6 states exists in this situation, thus we can only realize the multiplication between this 6 states. A sophisticatedly designed multiplication is shown in the figure ~\ref{fig:multiplier_35}, and we can verify that they are correct multipliers:

\begin{figure}[ht]
\centering
\includegraphics[width=0.55\textwidth]{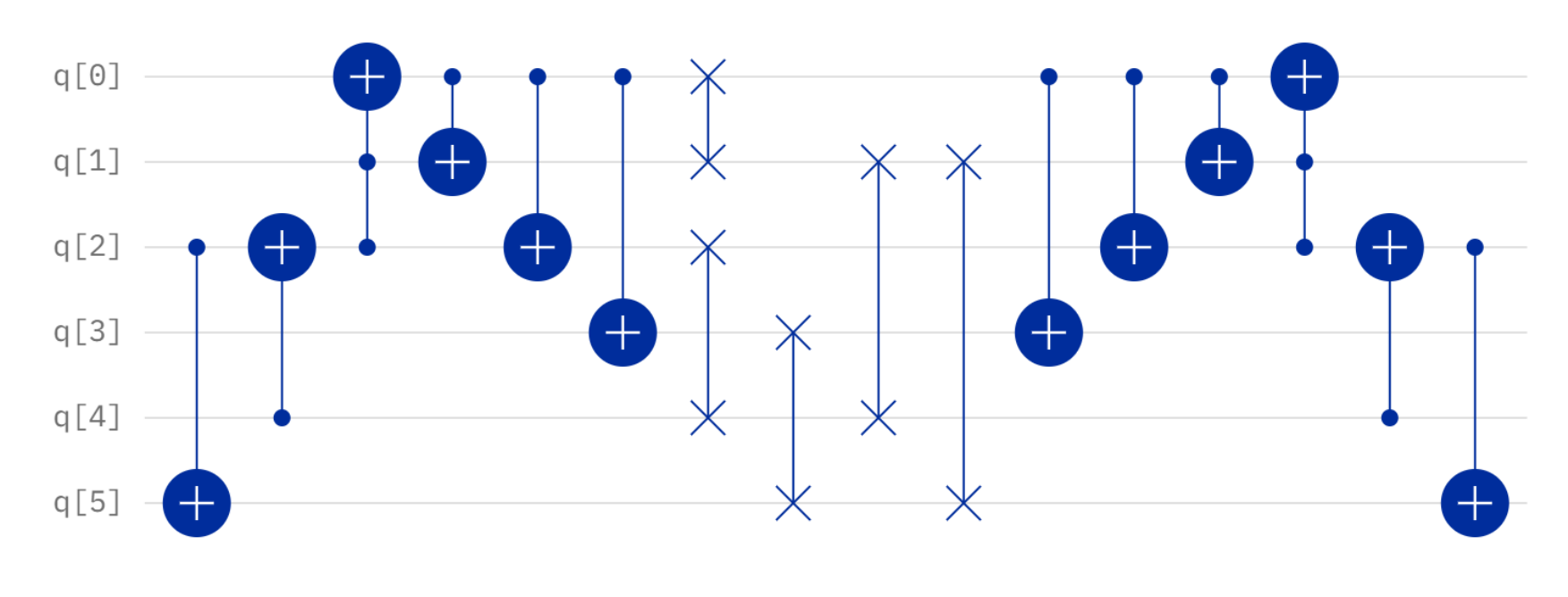}
\caption{Circuit implementation of the compiled multiplier for \( 4 \pmod{35} \). The implementation for \( 9 \pmod{35} \) is its reverse.}
\label{fig:multiplier_35}
\end{figure}

\begin{figure}[ht]
\centering
\includegraphics[width=0.45\textwidth]{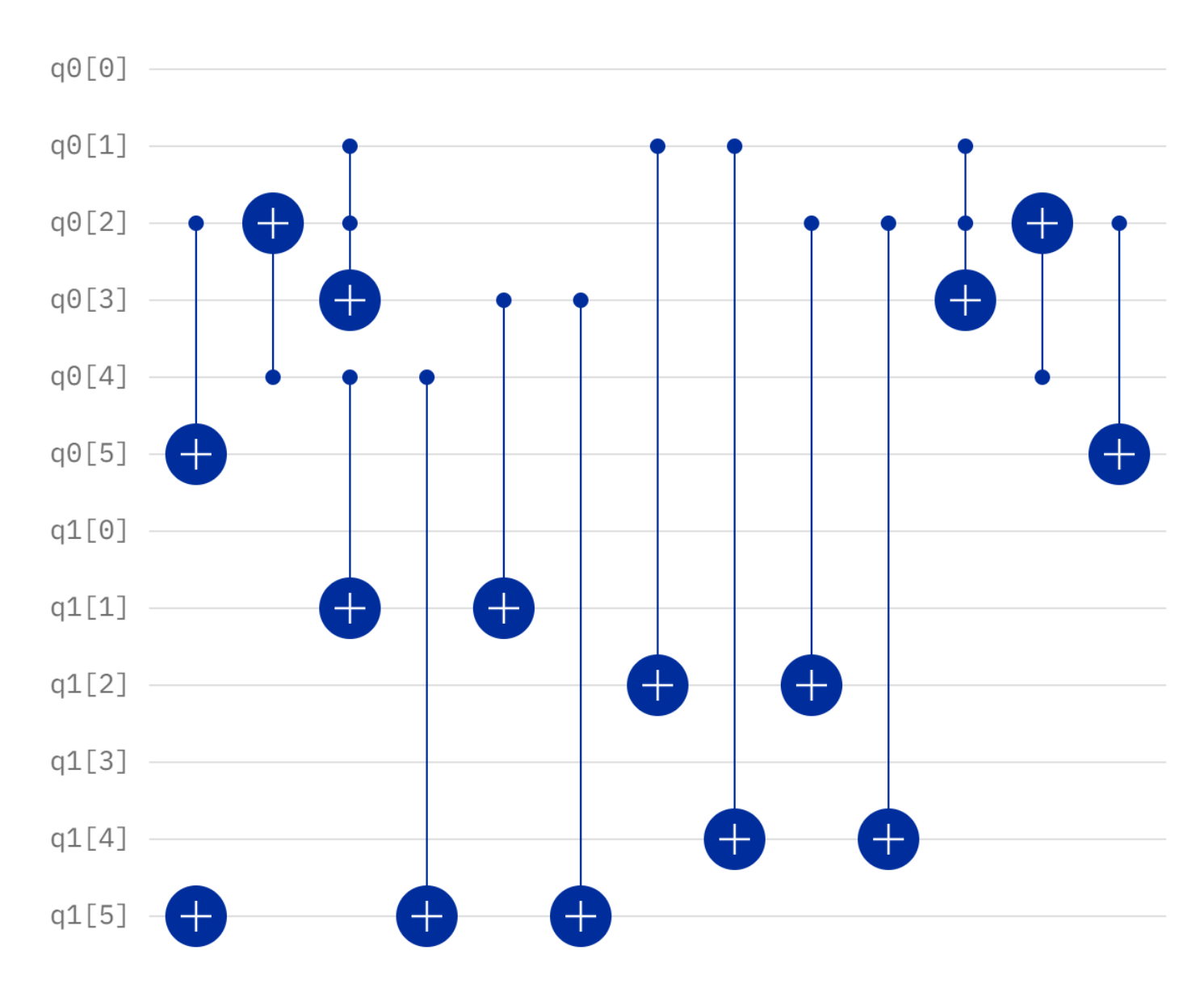}
\caption{Circuit implementation of the compiled out-of-place squaring for \( 35 \), with the \( q_0 \) register as input and the \( q_1 \) register as output.}
\label{fig:square_35}
\end{figure}

The squaring operation becomes complex, because even in the "compiled version", where only six target states are considered, the squaring operation remains irreversible. For example, while \(\ket{4}\xrightarrow{\text{square} \pmod{35}}\ket{16}\), simultaneously \(\ket{11}\xrightarrow{\text{square} \pmod{35}}\ket{16}\). Thus, such a mapping cannot be implemented in-place. It is necessary to introduce ancilla qubits for out-place storage and to erase intermediate states after the final computation is completed. In this example, to implement the squaring operation among the aforementioned six target states, the result is shown in Figure~\ref{fig:square_35}.

In practice, we can further simplify the quantum circuit in two ways:

\begin{enumerate}
\item \textbf{Recompilation of the six involved states}

It is noted that, after the modular exponentiation, only the control register undergoes the subsequent quantum Fourier transform, and the value of the computational register does not need to be measured. Thus, the specific quantum state of the computational register is irrelevant. In this process, the key aspect is to establish the correct quantum entanglement between the control and computational registers. This insight suggests that it may be possible to use fewer qubits to construct the computational register.
Considering that the computational register involves only six possible states, we can reduce the number of qubits in the computational register to three (with a Hilbert space dimension of \( 2^3 = 8 > 6 \)) using the following encoding scheme:

\begin{align*}
\ket{1} &\mapsto \ket{001}, \\
\ket{4} &\mapsto \ket{111}, \\
\ket{16} &\mapsto \ket{010}, \\
\ket{29} &\mapsto \ket{101}, \\
\ket{11} &\mapsto \ket{011}, \\
\ket{9} &\mapsto \ket{110}.
\end{align*}

Under this encoding, it can be readily verified that the transformations between these six states, defined as \( U_4 = \{ \ket{1} \to \ket{4}, \ket{4} \to \ket{16}, \ket{16} \to \ket{29}, \ket{29} \to \ket{11}, \ket{11} \to \ket{9}, \ket{9} \to \ket{1} \} \), and \( U_9 = U_4^{-1} \), can be implemented in a highly simplified form, as shown in Figure~\ref{fig:multi_sim}.

\begin{figure}[ht]
\centering
\includegraphics[width=0.4\textwidth]{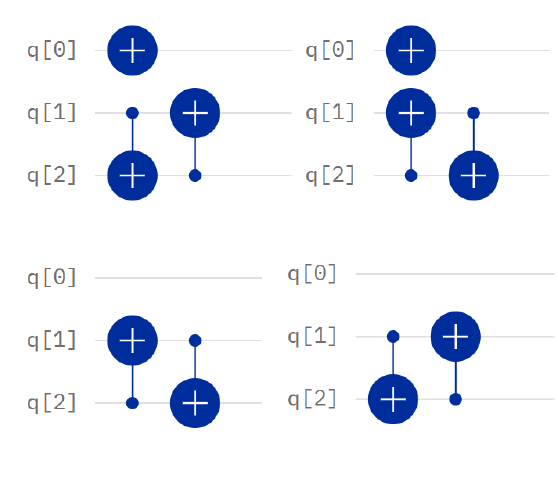}
\caption{Circuit implementation of the simplified multiplications for 4 and 9. The two subfigures in the top row illustrate the multiplications for \( 4 \pmod{35} \) and \( 9 \pmod{35} \), while the two subfigures in the bottom row depict the multiplications for \( 4^2 \pmod{35} \) and \( 4^4 \pmod{35} \) (corresponding to \( 9^4 \pmod{35} \) and \( 9^2 \pmod{35}  \)).}
\label{fig:multi_sim}
\end{figure}

Due to the irreversibility of the squaring operation, it still requires out-place computation. Compared to previous approaches, it has also been significantly simplified, as shown in Figure~\ref{fig:square_sim}, where \( q_0[0], \ldots, q_0[2] \) represents the original register, and \( q_1[0], \ldots, q_1[2] \) represents the new register. After all operations are completed, the computational results of all intermediate registers must be erased.

\begin{figure}[ht]
\centering
\includegraphics[width=0.25\textwidth]{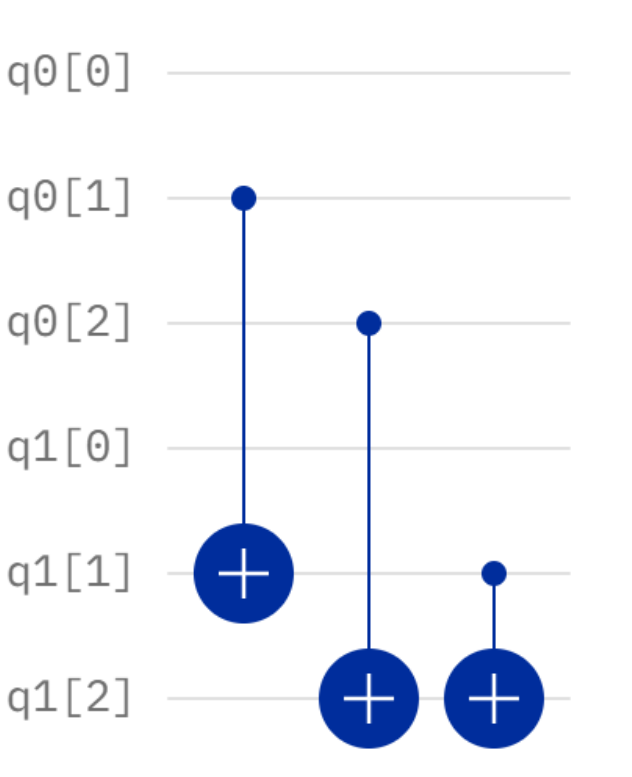}
\caption{Circuit implementation of the simplified square \(\mod{35}\).}
\label{fig:square_sim}
\end{figure}

    \item \textbf{Considering only the computational states actually encountered from the initial state}

    Furthermore, considering that the initial state of the computational register is always set to a definite state \(\ket{1} \mapsto \ket{001}\), the computation initially involves only a subset of states rather than all six states. During the application of the first few multipliers or squaring operations, certain controlled gates are certainly not triggered or are certainly triggered. This leads to further simplification of the circuit.

    For example, for the first multiplier, we need only consider a single state as input (since \(\ket{001}\) is the initial state). For the second multiplier, only two states need to be considered as input, and for the third multiplier or squaring operation, only three states need to be considered as input. The final result of the entire modular exponentiation module is illustrated in Figures~\ref{fig:circuit0} and~\ref{fig:circuit1}.

\end{enumerate}

Using the methods previously described, namely recompilation and simplified multiplier from initial state, we significantly reduce the circuit size.

\section{Appendix G: Circuits used for simulation and experiment}\label{G}

In this appendix, we present the full quantum circuits used in the simulations and experimental demonstrations discussed in the main text.

Figures~\ref{fig:circuit1} and~\ref{fig:circuit2} show the corresponding circuits for the square-and-multiply architecture, with and without our partial-uncomputation method, respectively. These circuits are used to compare the resource requirements and performance of our optimized construction against the original implementation. Figure~\ref{fig:circuit0} shows the complete circuit for the precomputation method, which is used in both classical simulations and experiments on real quantum hardware. 

\begin{figure}[ht]
\centering
\includegraphics[width=0.5\textwidth]{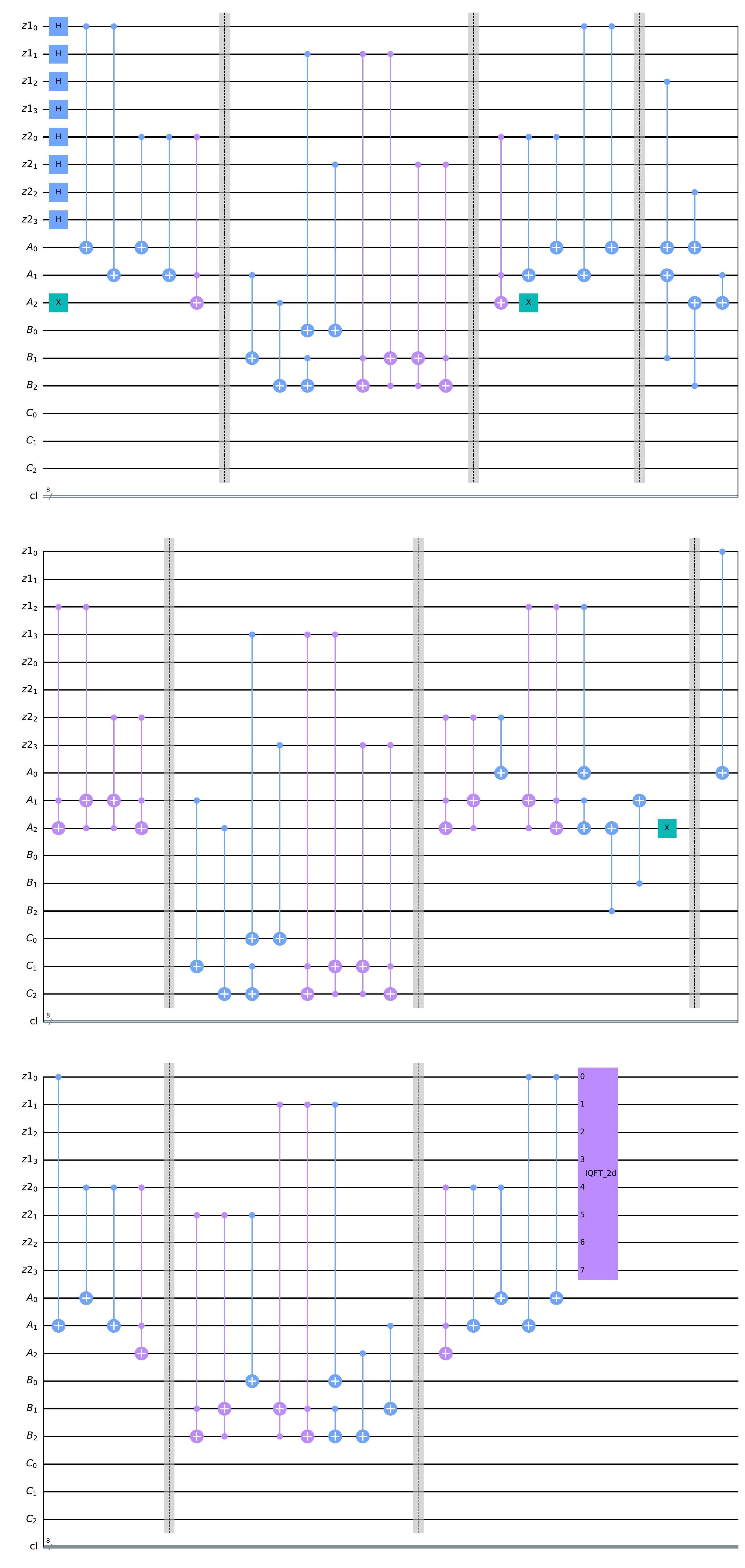}
\caption{Full circuit implementation for factoring \( N = 35 \) using the square-and-multiply architecture with our partial-uncomputation method. Each barrier separates different square-and-multiply blocks. The measurement operations are not shown in the figure.}
\label{fig:circuit1}
\end{figure}

\begin{figure}[ht]
\centering
\includegraphics[width=0.5\textwidth]{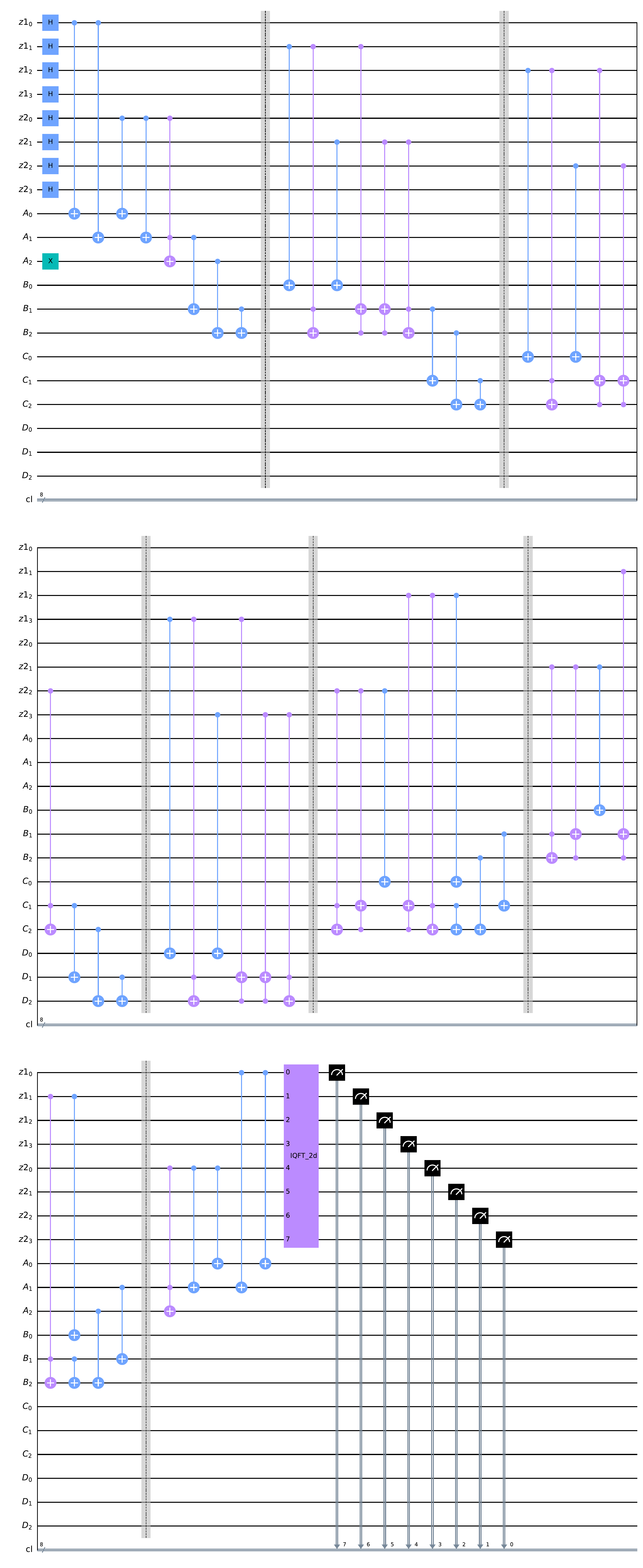}
\caption{Full circuit implementation for factoring \( N = 35 \) using the square-and-multiply architecture without our partial-uncomputation method. Each barrier separates different square-and-multiply blocks.}
\label{fig:circuit2}
\end{figure}

\begin{figure}[ht]
\centering
\includegraphics[width=0.8\textwidth]{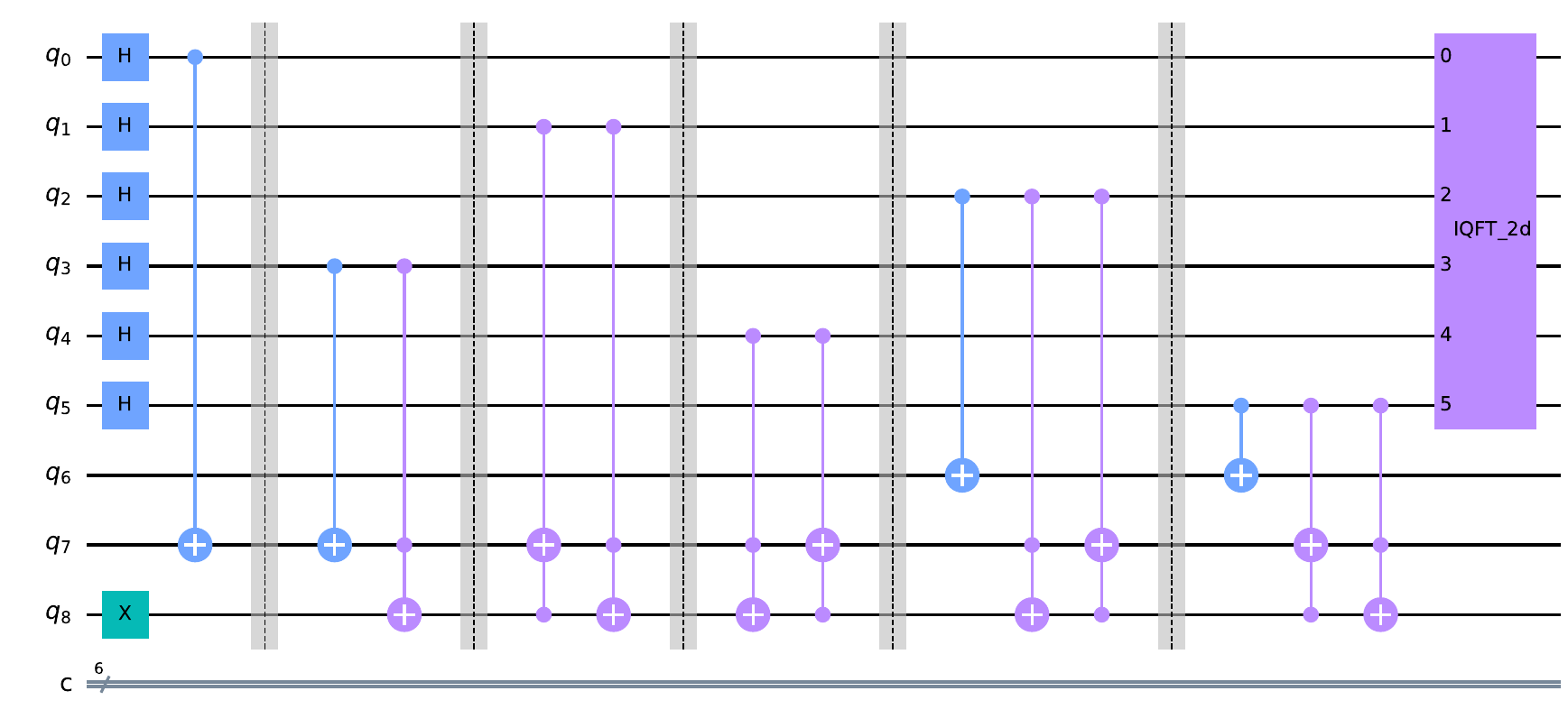}
\caption{Full circuit implementation for factoring \( N = 35 \) using the precomputation method. Each barrier separates different multiplier blocks. The measurement operations on qubits \(q_0\) to \(q_5\) are not shown in the figure.}
\label{fig:circuit0}
\end{figure}

\end{document}